\documentclass[journal]{IEEEtran}
\usepackage{amsmath,subfigure,multirow}
\usepackage{graphicx,amssymb,lineno,bm,subfigure}
\usepackage{algorithm,color}
\usepackage{algorithmic}
\usepackage{cite}
\usepackage{array}
\usepackage{float}
\usepackage{subeqnarray}
\usepackage{cases}
\usepackage{url}

\newtheorem{thm}{\textbf{Theorem}}
\newtheorem{rmk}{\textbf{Remark}}
\newtheorem{lma}{\textbf{Lemma}}
\newtheorem{defi}{\textbf{Definition}}
\newtheorem{prop}{\textbf{Proposition}}
\newtheorem{corol}{\textbf{Corollary}}

\begin{document}

\title{A Lyapunov Optimization Approach for Green Cellular Networks with Hybrid Energy Supplies\\}

\author{\normalsize \authorblockN{Yuyi~Mao,~\IEEEmembership{Student Member,~IEEE},
Jun~Zhang,~\IEEEmembership{Senior Member,~IEEE}, and Khaled~B. Letaief,~\IEEEmembership{Fellow,~IEEE}}
\thanks{The authors are with the Department of Electronic and Computer Engineering, the Hong Kong University
of Science and Technology, Clear Water Bay, Kowloon, Hong Kong. E-mail: \{ymaoac, eejzhang, eekhaled@ust.hk\}.}

}



\maketitle \vspace{0.3cm}
\begin{abstract}
Powering cellular networks with renewable energy sources via energy harvesting (EH) has recently been proposed as a promising solution for green networking. However, with intermittent and random energy arrivals, it is challenging to provide satisfactory quality of service (QoS) in EH networks. To enjoy the greenness brought by EH while overcoming the instability of the renewable energy sources, hybrid energy supply (HES) networks that are powered by both EH and the electric grid have emerged as a new paradigm for green communications. In this paper, we will propose new design methodologies for HES green cellular networks with the help of Lyapunov optimization techniques. The network service cost, which addresses both the grid energy consumption and achievable QoS, is adopted as the performance metric, and it is optimized via base station assignment and power control (BAPC). Our main contribution is a low-complexity online algorithm to minimize the long-term average network service cost, namely, the \emph{Lyapunov optimization-based BAPC} (LBAPC) algorithm. One main advantage of this algorithm is that the decisions depend only on the instantaneous side information without requiring distribution information of channels and EH processes. To determine the network operation, we only need to solve a deterministic per-time slot problem, for which an efficient inner-outer optimization algorithm is proposed. Moreover, the proposed algorithm is shown to be asymptotically optimal via rigorous analysis. Finally, sample simulation results are presented to verify the theoretical analysis as well as validate the effectiveness of the proposed algorithm.
\end{abstract}

\begin{keywords}
Green communications, energy harvesting, hybrid energy supply, base station assignment, power control, QoS, Lyapunov optimization.
\end{keywords}

\section{Introduction}
\IEEEPARstart{T}HE continuous growth of wireless applications combined with the proliferation of smart mobile devices has resulted in an unprecedented growth of wireless data traffic, which has contributed to a dramatic increase in the energy consumption and carbon emission of the Information and Communication Technology (ICT) sector. It is estimated that the annual carbon emissions and electric power consumption of the ICT industry will reach up to 235 Mto \cite{Fehske11} and 414 TWh \cite{Lamber12}, respectively, in 2020. Heterogeneous and small cell networks (HetSNets) provide an energy-efficient paradigm to improve the network capacity, and thus have been regarded as one of the most promising solutions for realizing green radio \cite{IHwang13,CLi14}. However, with the dense deployment of base stations (BSs) in HetSNets, the overall energy consumption and carbon footprint will still be high \cite{CLi13}. Consequently, it is urgent to seek alternative green energy sources for wireless networks.

The recent advances of energy harvesting (EH) technologies enable the BSs with EH components to capture ambient recyclable energy, e.g., solar radiation and wind energy, which is promising to achieve green networking \cite{Ulukus15,YMao1502}. By introducing EH capabilities to the next-generation cellular networks, potentially 20\% of their $CO_{2}$ emission can be reduced \cite{GPiro13}. Nevertheless, since the surrounding harvestable energy depends highly on environmental factors such as location and weather, the harvested energy is unstable by nature. As a result, it is challenging to maintain satisfactory quality of service (QoS) if communication nodes are solely powered by the harvested renewable energy. To enjoy the environmental friendliness of EH, and also to overcome the unreliability of the renewable energy sources, wireless networks with a hybrid energy supply (HES), where EH and the electric grid coexist, will be an ideal solution. While HES networks have attracted recent attention, they also bring new design challenges. In particular, communication protocols developed either for conventional grid-powered cellular networks or EH systems cannot take the full benefits of the heterogeneous energy sources in HES networks. In this paper, we shall propose new design methodologies for HES wireless networks, which will provide valuable guidelines for developing green cellular networks supported by renewable energy sources in the near future.

\subsection{Related Works and Motivations}

EH communications have attracted significant attention from academia in recent years. It was revealed that, with either the save-then-transmit protocol or the best-effort protocol, the capacity of the point-to-point \emph{additive white Gaussian noise} (AWGN) channel can be achieved if the transmitter is powered by EH \cite{Ozel12}. This result indicates the benefits of EH communications from the information theoretical perspective. However, as the harvested energy is intermittent and sporadic, on one hand, energy management in EH systems should be based on the \emph{channel side information} (CSI) as in conventional systems, but on the other hand, it should be adaptive to the \emph{energy side information} (ESI). With non-causal side information (SI)\footnote{`Causal SI' refers to the case that, at any time instant, only the past and current SI is known, while `non-causal SI' means that the future SI is also available.}, including the CSI and ESI, the maximum throughput of point-to-point EH fading channels can be achieved by the \emph{directional water-filling} (DWF) algorithm \cite{Ozel11}. The study was later extended to broadcast channels \cite{JYang1202}, multiple access channels \cite{JYang1204}, and cooperative communications systems \cite{CHuang13,YLuo13}. Besides these, scenarios with more practical assumptions on SI have been investigated in \cite{LHuang13,Blasco13,ZWang14}.

Transmission protocols for HES systems have also been recently studied, where the main focus is on point-to-point systems. Given the grid energy budget, a \emph{two-stage} DWF algorithm was proposed in \cite{JGong13} to achieve the optimal throughput with non-causal SI at the transmitter. {A similar problem was investigated in \cite{PeterHe13}, where a low-complexity \emph{recursive geometric water-filling} algorithm was derived.} In \cite{XKang1408}, power allocation strategies for weighted energy cost minimization in a point-to-point HES link were proposed. {For wireless links with hybrid energy sources, optimal power allocation policies to minimize the non-harvested energy consumption with delay-constrained data traffic requirement were proposed in \cite{IAhmed13}.} And resource allocation policies to maximize the energy efficiency in HES OFDMA systems were developed in \cite{DNg13}. To realize green networking, more recently, powering cellular networks with hybrid energy supplies has been proposed \cite{THan1308,THan1312,JGong14,JXu14,YMao1502}. The design in the network setting becomes more challenging since more decisions should be made, and more SI will be needed. To save the grid energy consumption, the green energy utilization optimization problem was solved in \cite{THan1308} assuming full ESI was available, and a green energy and latency-aware user association scheme was proposed in \cite{THan1312}.  In \cite{JGong14}, a sleep control scheme for HES networks was developed, while joint energy cooperation and communication cooperation for HES \emph{coordinated multi-point} (CoMP) systems was proposed in \cite{JXu14}.

To simplify the design, previous studies on HES networks either ignore the accumulation of harvested energy at BSs \cite{THan1312,JXu14}, or assume non-causal ESI is available \cite{THan1308,JGong14}, which cannot fully capture the intermittency and randomness of EH. In general, for more practical online scenarios with causal SI, the optimal transmission policies remain unknown. For many online cases, the design problem can be formulated as a \emph{Markov Decision Process} (MDP) problem, and thus can be solved in principle. However, due to the huge dimension of system states in HES networks, the complexity of the MDP solutions is unacceptable. Although heuristic policies can be developed, they generally do not have any performance guarantees. Motivated by these limitations in existing works, in this paper, we will investigate how to design practical online transmission protocols for HES cellular networks with desirable properties such as low complexity and theoretical performance guarantees. Specifically, Lyapunov optimization will be used as the main tool. Such techniques have a long history in the field of discrete stochastic processes and Markov chains \cite{SMeyn09}.  Moreover, it has been one of the most important methods for delay-aware resource control problems in wireless systems \cite{YCui12TIT}, while application in EH networks was first proposed by Huang \emph{et al.} \cite{LHuang13}. The algorithms developed from the Lyapunov optimization techniques enjoy various attractive properties, e.g., little requirement of prior knowledge, low computational complexity, and quantifiable worst-case performance, which make them a good fit for HES networks.

\subsection{Contributions}

In this paper, we will develop effective online algorithms to optimize green cellular networks powered by hybrid energy sources based on Lyapunov optimization techniques. Our major contributions are summarized as follows:

\begin{itemize}

\item
We consider a multi-user HES cellular network, and a network service cost that incorporates both the grid energy consumption and achievable QoS is adopted as the performance metric. The \emph{network service cost minimization} (NSCM) problem, which is an intractable high-dimension Markov decision problem, is first formulated assuming causal SI. A modified NSCM problem with tightened battery output power constraints is then proposed, which will assist the algorithm design based on Lyapunov optimization techniques.

\item
A low-complexity online \textbf{L}yapunov optimization-based \textbf{b}ase station \textbf{a}ssignment and \textbf{p}ower \textbf{c}ontrol (LBAPC) algorithm is proposed for the modified NSCM problem, which also provides a feasible solution to the original problem. In each time slot, the network operation only depends on the optimal solution of a deterministic optimization problem, which can be solved efficiently by a proposed inner-outer optimization algorithm.

\item
Performance analysis for the LBAPC algorithm is conducted. It is shown that the proposed algorithm can achieve asymptotically optimal performance of the original NSCM problem by tuning a set of control parameters. Moreover, it does not require statistical information of the involved stochastic processes including both the channel and EH processes, which makes it also applicable in unpredictable environments.

\item
Simulation results are provided to verify the theoretical analysis, especially the asymptotic optimality of the LBAPC algorithm. Moreover, the effectiveness of the proposed policy is demonstrated by comparison with a greedy transmission scheme. It will be shown that the LBAPC algorithm not only achieves significant performance improvement in terms of the network service cost, but also greatly reduces both the network grid energy consumption as well as the packet drop ratio. Moreover, it can more efficiently utilize the available spectrum.
\end{itemize}

The organization of this paper is as follows. In Section II, we introduce the system model. The NSCM problem is formulated in Section III. The LBAPC algorithm for the NSCM problem is proposed in Section IV and its performance analysis is conducted in Section V. We present the simulation results in Section VI and conclude this paper in Section VII.

\section{System model}
\begin{figure}[h]
\begin{center}
    \label{GEMP}
   \includegraphics[width=0.45\textwidth]{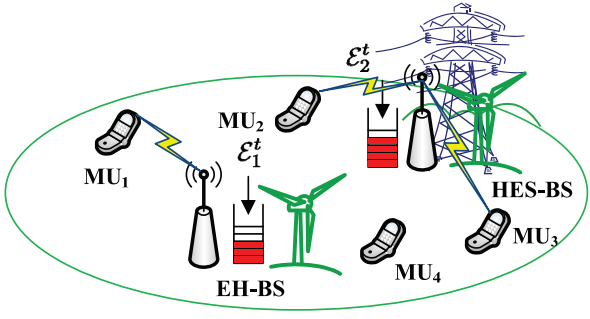}
\end{center}
\vspace{-10pt}
\caption{A multi-user HES wireless network with one EH-BS, one HES-BS and four mobile users.}
\end{figure}
We consider a multi-user HES wireless network with an EH-BS ($\mathcal{B}_{1}$), an HES-BS ($\mathcal{B}_{2}$), and $K$ mobile users (MUs), as shown in Fig. 1. The index set of the MUs is denoted as $\mathcal{K}=\{1,\cdots,K\}$. Two BSs coordinate to serve the MUs. {The EH-BS, which can be a small cell to increase system capacity \cite{YMao1502}, is equipped with an EH component and powered purely by the harvested renewable energy, with the maximum transmit power given by $p_{\mathcal{B}_{1}}^{\max}$. The HES-BS is not only mounted with an EH component, but also connected to the electric grid, i.e., it can utilize both the harvested energy and the grid energy, with the maximum transmit power denoted as $p_{\mathcal{B}_{2}}^{\max}$. This HES-BS may serve as a macro-BS to guarantee coverage. Both $p_{\mathcal{B}_{1}}^{\max}$ and $p_{\mathcal{B}_{2}}^{\max}$ are assumed to be bounded. In this paper, the HES-BS acts as the decision center and collects all the SI needed for decision making. For ease of reference, we list the key notations of our system model in TABLE I.}

Time is slotted, and denote $\tau$ as the time slot length  and $\mathcal{T}=\{0,1,\cdots\}$ as the set of time slot indices. Downlink transmission is considered. Particularly, we assume that at the beginning of each time slot, a packet with $R$ bits arrives from upper layers at the BS side for each $\mathrm{MU}$. These packets may be transmitted via either the EH-BS or the HES-BS in the following time slot. It may also happen that neither of the BSs is able to deliver the packet. BS assignment and power control will be adopted to optimize the network. Denote $I_{j,k}^{t}\in\{0,1\}$, with $j\in\{\mathcal{B}_{1},\mathcal{B}_{2},D\}$, as the BS assignment indicator for $\mathrm{MU}_{k}$ in the $t$th time slot, where $I_{\mathcal{B}_{1},k}^{t}=1$ ($I_{\mathcal{B}_{2},k}^{t}=1$) indicates that the EH-BS (HES-BS) is assigned to serve $\mathrm{MU}_{k}$, and $I_{D,k}^{t}=1$ means neither of the BSs will transmit the packet, i.e., the packet is dropped. These indicators are subjected to the following operation constraint:
\begin{equation}
\sum_{j\in\{\mathcal{B}_{1},\mathcal{B}_{2},D\}}I_{j,k}^{t}=1,\forall k\in\mathcal{K},t\in\mathcal{T}.
\label{operconst}
\end{equation}

\begin{table}[ht]
\center \protect
\caption{Summary of Key Notations}
\begin{tabular}{ll}
\hline
{\textbf{Notation}} & {\textbf{Description}}  \tabularnewline
\hline
{$\mathcal{B}_{1}$ ($\mathcal{B}_{2}$)} &The EH-BS (HES-BS)  \tabularnewline
{$\mathcal{K}$} &Index set of the MUs \tabularnewline
{$\mathcal{T}$} &Index set of the time slots \tabularnewline
{$R$} &Packet size \tabularnewline
{$\{I_{j,k}^{t}\}$} &BS assignment indicator for $\mathrm{MU}_k$ in time slot $t$ \tabularnewline
{$p_{\mathcal{B}_{1}}^{\max}$ ($p_{\mathcal{B}_{2}}^{\max}$)} &Maximum transmit power at the EH-BS (HES-BS)\tabularnewline
\multirow{2}{*}{$p_{\mathcal{B}_{1},k}^{t}$ ($p_{\mathcal{B}_{2},k}^{t}$)} &Transmit power of the EH-BS (HES-BS) for  \tabularnewline &$\mathrm{MU}_k$ in time slot $t$\tabularnewline
\multirow{2}{*}{$p_{H_{1},k}^{t}$ ($p_{H_{2},k}^{t}$)} &Power consumption of harvested energy at the   \tabularnewline
& EH-BS (HES-BS) for $\mathrm{MU}_k$ in time slot $t$ \tabularnewline
\multirow{2}{*}{$p_{G,k}^{t}$} &Power consumption of grid energy at the  \tabularnewline &HES-BS for $\mathrm{MU}_{k}$ in time slot $t$ \tabularnewline
\multirow{2}{*}{$B_{1}^{t}$ ($B_{2}^{t}$)} &Battery energy level at the EH-BS (HES-BS) \tabularnewline &in time slot $t$\tabularnewline
\multirow{2}{*}{$\mathcal{E}_{1}^{t}$ ($\mathcal{E}_{2}^{t}$)} &Harvestable energy at the EH-BS (HES-BS) \tabularnewline &in time slot $t$ \tabularnewline
{$E_{H_{1}}^{\max}$ ($E_{H_{2}}^{\max}$)} &Maximum value of $\mathcal{E}_{1}^{t}$ ($\mathcal{E}_{2}^{t}$) \tabularnewline
\multirow{2}{*}{$e_{1}^{t}$ ($e_{2}^{t}$)} &Harvested energy at the EH-BS (HES-BS) \tabularnewline &in time slot $t$ \tabularnewline
\multirow{2}{*}{$h_{\mathcal{B}_{1},k}^{t}$ ($h_{\mathcal{B}_{2},k}^{t}$)} &Channel gain from the EH-BS (HES-BS) \tabularnewline &to $\mathrm{MU}_{k}$ in time slot $t$ \tabularnewline
\multirow{2}{*}{$N_{\mathcal{B}_{1}}$ ($N_{\mathcal{B}_{2}}$)} &Number of available orthogonal channels \tabularnewline &at the EH-BS (HES-BS) \tabularnewline
\multirow{2}{*}{$\varphi_{G}$ ($\varphi_{D}$)} &Cost incurred by per Joule of grid energy \tabularnewline &consumption (per packet drop) \tabularnewline
{$w_{G}$ ($w_{D}$)} &The weight of the grid energy cost (packet drop cost) \tabularnewline
\hline
\end{tabular}
\end{table}

The transmit powers of the EH-BS and the HES-BS for $\mathrm{MU}_{k}$ in time slot $t$ are denoted as $p_{\mathcal{B}_{1},k}^{t}$ and $p_{\mathcal{B}_{2},k}^{t}$, respectively. As $p_{\mathcal{B}_{1},k}^{t}$ originates from the harvested energy at the EH-BS, for convenience it is also denoted as $p_{H_{1},k}^{t}$, while $p_{\mathcal{B}_{2},k}^{t}$ consists of both the harvested energy, denoted as $p_{{H}_{2},k}^{t}$, and the grid energy, denoted as $p_{G,k}^{t}$, i.e., $p_{\mathcal{B}_{2},k}^{t}=p_{H_{2},k}^{t}+p_{G,k}^{t}$. In this work, the energy consumed for purposes other than transmission, e.g., cooling and baseband signal processing, is neglected.

The EH processes are modeled as successive energy packet arrivals, i.e., at the beginning of each time slot, energy packets with $\mathcal{E}_{1}^{t}$ and $\mathcal{E}_{2}^{t}$ arrive at the EH-BS and the HES-BS, respectively. We assume $\mathcal{E}_{1}^{t}$'s ($\mathcal{E}_{2}^{t}$'s) are independent and identically distributed (i.i.d.) among different time slots with the maximum value $E_{H_{1}}^{\max}$ ($E_{H_{2}}^{\max}$). {Although the i.i.d. EH model is idealized, it captures the intermittent nature of the EH processes, and thus it has been widely adopted in the literature, e.g., \cite{Ozel12,LHuang13,XKang1408}.} In each time slot, part of the arrived energy, denoted as $e_{j}^{t}$, satisfying
\begin{equation}
0\leq e_{j}^{t}\leq \mathcal{E}_{j}^{t},j\in\{1,2\},t\in \mathcal{T},
\label{harvestconst}
\end{equation}
will be harvested and stored in a battery, and it will be available for transmission from the next time slot. We start by assuming that the battery capacity is sufficiently large. Later we will show that by picking the values of $e_{j}^{t}$'s, the battery energy levels are deterministically upper-bounded under the proposed algorithm, thus we only need finite-capacity batteries in the actual implementation. {More importantly, including $e_{j}^{t}$'s as variables in the optimization facilitates the derivation and performance analysis of the proposed Lyapunov optimization-based algorithm, which will be elaborated in the following sections.} Similar techniques were adopted in previous studies, such as \cite{LHuang13} and \cite{Laksh14}. Denote the battery energy levels of the EH-BS and the HES-BS at the beginning of time slot $t$ as $B_{1}^{t}$ and $B_{2}^{t}$, respectively. Without loss of generality, we assume $B_{j}^{0}=0$ and $B_{j}^{t}< +\infty,j=1,2$. Since the renewable energy that has not yet been harvested can not be utilized, the following energy causality constraint should be satisfied:
\begin{equation}
\sum_{k\in\mathcal{K}}p_{H_{j},k}^{t}\tau \leq B_{j}^{t}< +\infty,j\in\{1,2\},\forall t\in\mathcal{T}.
\label{EHcausality}
\end{equation} Thus, $B_{j}^{t}$ evolves according to
\begin{equation}
B_{j}^{t+1}=B_{j}^{t}-\sum_{k\in\mathcal{K}} p_{H_{j},k}^{t} \tau+e_{j}^{t}, j\in\{1,2\},t\in\mathcal{T}.
\label{batterydynamics}
\end{equation}

The BSs will serve the MUs with multiple orthogonal channels, e.g., by adopting OFDMA as in the LTE standard \cite{JZhangbook}. {Static and orthogonal spectrum allocation is adopted for the two BSs, which is a popular scheme for tiered cellular networks \cite{Chandra09,WCheung12}.} Specifically, we assume $N_{\mathcal{B}_{1}}+N_{\mathcal{B}_{2}}$ orthogonal channels with equal bandwidth $w$ Hz are available, where $N_{\mathcal{B}_{1}}\geq 1$ of them are allocated for the EH-BS while the remaining $N_{\mathcal{B}_{2}}\geq 1$ channels are reserved for the HES-BS. In this paper, $N_{\mathcal{B}_{1}}$ and $N_{\mathcal{B}_{2}}$ are fixed and pre-determined. {It is worthwhile to note that, $N_{\mathcal{B}_{1}}$ and $N_{\mathcal{B}_{2}}$ can be further optimized with the assistance of \emph{orthogonal access spectrum allocation} strategies based on various system parameters, such as the EH conditions at both BSs and the network traffic demand \cite{Chandra09,WCheung12}. This is beyond the scope of this paper and will be handled in our future works.} As a result, at each time slot, the EH-BS and HES-BS can serve at most $N_{\mathcal{B}_{1}}$ and $N_{\mathcal{B}_{2}}$ MUs, respectively.  So the following channel assignment constraint should be met:
\begin{equation}
\sum_{k\in\mathcal{K}}I_{j,k}^{t}\leq N_{j}, j\in\{\mathcal{B}_{1},\mathcal{B}_{2}\}, t\in\mathcal{T}.
\label{channelconst}
\end{equation}

The channels are assumed to be flat fading within the considered bandwidth and the channel gains are i.i.d. among different time slots. We denote the channel gains from the EH-BS (HES-BS) to $\mathrm{MU}_{k}$ as $h_{\mathcal{B}_{1},k}^{t}$ ($h_{\mathcal{B}_{2},k}^{t}$). {For simplicity, we further assume the channel gains from the BSs to the MUs are statistically independent and identical.} Thus, $h_{\mathcal{B}_{j},k}^{t}\sim F_{\mathcal{B}_{j}}\left(x\right),j=1,2,\forall k\in\mathcal{K}$, where $F_{\mathcal{B}_{j}}\left(x\right)$ denotes the cumulative distribution functions of $h_{\mathcal{B}_{j},k}^{t}$. {As is the case in most of the existing works on EH communications, error-free channel estimation and feedback are assumed. Thus, perfect CSI is available at the transmitters. Therefore, the results in this paper can serve as a design guideline and a performance upper bound for cases where only imperfect CSI is available.} Consequently, if the EH-BS (or HES-BS) serves $\mathrm{MU}_{k}$ in time slot $t$, the throughput is given by $r\left(h_{\mathcal{B}_{1},k}^{t},p_{\mathcal{B}_{1},k}^{t}\right)$ (or $r\left(h_{\mathcal{B}_{2},k}^{t},p_{\mathcal{B}_{2},k}^{t}\right)$), where $r\left(h,p\right)=w\tau \log_{2}\left(1+\frac{h p}{\sigma}\right)$ is the Shannon-Hartley formula and $\sigma$ is the noise power at the receiver.

\section{Problem formulation}
In this section, we will first introduce the performance metric, namely, the network service cost. Then the network service cost minimization (NSCM) problem will be formulated, and its unique technical challenges will be identified.

\subsection{The Network Service Cost Minimization Problem}
As mentioned in Section II, it may happen that some data packets can not be successfully delivered to the corresponding MUs. For instance, the EH-BS does not have enough harvested energy while the channel from the HES-BS is in deep fading, or all the available channels are occupied. In such circumstances, neither of the BSs is capable of transmitting the data packet, which induces a \emph{packet drop cost} $\varphi_{D}$. In some real-time applications, this
packet may indeed be dropped, while for other applications, this will increase the delay. On the other hand, due to the intermittent and sporadic nature of energy harvesting, the HES-BS will need to use the grid energy for transmission from time to time, which incurs a \emph{grid energy cost} $\varphi_{G}$ per Joule. {Minimizing the grid energy consumption and maximizing the provided QoS are two important design objectives for HES networks.} Thus, it is desirable to minimize both types of costs in order to optimize the system. In this paper, we will adopt the \emph{service cost} as the performance metric, which is the weighted sum of the grid energy cost and packet drop cost. This metric was introduced in \cite{YMao15}, where it was shown to be capable of adjusting the tradeoff between the grid energy consumption and the achievable QoS in HES networks. We will use the \emph{network service cost} (NSC) as the performance metric for the studied HES network, which is the total service cost for all MUs and defined for the $t$th time slot as
\begin{equation}
\mathrm{NSC}^{t}\triangleq\sum_{k\in\mathcal{K}}\left(w_{G}\varphi_{G} p_{G,k}^{t}\tau+w_{D}\varphi_{D} I_{D,k}^{t}\right),
\label{NSC}
\end{equation}
where the first term presents the grid energy cost, the second term stands for the packet drop cost, and $w_{G}$ and $w_{D}$ are the weights of the grid energy cost and packet drop cost, respectively. {When $w_{G} \gg w_{D}$, the network is grid energy sensitive. On the other hand, when $w_{D} \gg w_{G}$, the network places more emphasis on the successful packet delivery, i.e., addresses on QoS more. Without loss of generality, we assume $\varphi_{G}$, $\varphi_{D}$, $w_{G}$ and $w_{D}$ are positive and bounded.
\begin{rmk}
Since the channel gains from the EH-BS (HES-BS) to different MUs are assumed to be i.i.d. in this paper, in order to guarantee fairness, the weights of the packet drop cost and grid energy cost are set to be the same for all MUs in (\ref{NSC}), i.e., $w_{G}$ and $w_{D}$. For cases where the channel gains from the BSs to the MUs are not statistically identical, the weights can be adjusted accordingly, and the proposed Lyapunov optimization approach still applies.
\end{rmk}}

Our design objective will be to minimize the long-term average network service cost, which addresses both the network grid energy consumption and the achievable QoS, i.e., the percentage of successfully transmitted data packets of the MUs. For each MU, the following QoS constraint should be satisfied:
\begin{equation}
\sum_{j\in\{\mathcal{B}_{1},\mathcal{B}_{2}\}}I_{j,k}^{t}r\left(h_{j,k}^{t},p_{j,k}^{t}\right)\geq \left(1-I_{D,k}^{t}\right)R, k\in\mathcal{K},t\in\mathcal{T},
\label{QoSconst}
\end{equation}
i.e., if it is decided to transmit the packet, the throughput should be greater than the packet size.
Consequently, the NSCM problem can be formulated as
\begin{align}
&\mathcal{P}1:\min_{\mathbf{I}^{t},\mathbf{p}^{t},\mathbf{e}^{t}} \lim_{T\rightarrow +\infty} \frac{1}{T}\sum_{t=0}^{T-1}\mathbb{E}\left[\sum_{k=1}^{K}\left(\tilde{\varphi}_{G} p_{G,k}^{t}\tau+\tilde{\varphi}_{D} I_{D,k}^{t}\right)\right]\nonumber\\
&\ \ \ \ \ \ \ \mathrm{s.t.}\ \ \ (\ref{operconst})-(\ref{EHcausality}), (\ref{channelconst}), (\ref{QoSconst}) \nonumber\\
&\ \ \ \ \ \ \ \ \ \ \ \ \ \ 0\leq p_{j,k}^{t}\leq I_{j,k}^{t}p_{j}^{\max},\forall j\in\{\mathcal{B}_{1},\mathcal{B}_{2}\},t\in\mathcal{T},k\in\mathcal{K}\label{NsNp}\\
&\ \ \ \ \ \ \ \ \ \ \ \ \ \sum_{k\in\mathcal{K}}p_{j,k}^{t}\leq p_{j}^{\max},\forall j\in\{\mathcal{B}_{1},\mathcal{B}_{2}\},t\in\mathcal{T}\label{PPC}\\
&\ \ \ \ \ \ \ \ \ \ \ \ \ \ I_{j,k}^{t}\in\{0,1\},\forall j\in\{\mathcal{B}_{1},\mathcal{B}_{2},D\},t\in\mathcal{T},k\in\mathcal{K}\label{01indicator}\\
&\ \ \ \ \ \ \ \ \ \ \ \ \ \ p_{H_{1},k}^{t},p_{H_{2},k}^{t},p_{G,k}^{t}\geq 0, \forall t\in\mathcal{T},k\in\mathcal{K},\label{NonnegaPower}
\end{align}
where $\tilde{\varphi}_{G}\triangleq w_{G}\varphi_{G}$, $\tilde{\varphi}_{D}\triangleq w_{D}\varphi_{D}$, $\mathbf{I}^{t}=\left[\{I_{\mathcal{B}_{1},k}^{t}\},\{I_{\mathcal{B}_{2},k}^{t}\},\{I_{D,k}^{t}\}\right]$, $\mathbf{p}^{t}=\left[\{p_{H_{1},k}^{t}\},\{p_{H_{2},k}^{t}\},\{p_{G,k}^{t}\}\right]$ and $\mathbf{e}^{t}=\left[e_{1}^{t},e_{2}^{t}\right]$. Thus we need to determine the BS assignment indicators $\mathbf{I}^t$, the power allocation $\mathbf{p}^t$, and the harvested energy at both BSs, i.e., $e_{1}^{t}$ and $e_{2}^{t}$. In $\mathcal{P}1$, (\ref{NsNp}) indicates that there will be no power allocated if an MU is not being served. The peak transmit power constraint and the power non-negativity constraint are imposed in (\ref{PPC}) and (\ref{NonnegaPower}), respectively. Moreover, the zero-one indicator constraint for the BS assignment indicators is represented by (\ref{01indicator}).

\subsection{Problem Analysis}
{In the considered HES network, the system state includes the battery energy levels, harvestable energy at each BS and the channel states, and the action is the energy harvesting as well as the BS assignment and power control. It can be checked easily that the allowable action set in each time slot only depends on the current system state. Also, the state transition is Markovian, which depends on the current system state and action, and is irrelevant with the state and action history. Besides, the objective function in $\mathcal{P}1$ is the long-term average network service cost. Therefore, $\mathcal{P}1$ is an infinite-horizon Markov decision process (MDP) problem.} In principle, the optimal solution of $\mathcal{P}1$ can be obtained by standard MDP algorithms, e.g., the \emph{value iteration algorithm} and \emph{linear programming reformulation} approach \cite{BertsekasDP}. Nevertheless, in both algorithms, we need to use finite states to characterize the system for practical implementation. For example, when $K=4$, if we use $M=10$ states to quantize the energy level at each BS, $E=5$ states to describe the harvestable energy, and $L=10$ states to represent the channel gain, overall, there will be $M^{2}E^{2}L^{2K}=2.5\times 10^{11}$ possible system states. For the \emph{value iteration algorithm}, it will take unacceptably long time to converge to the optimal value function, while for the \emph{linear programming reformulation}, a large-scale linear programming (LP) problem with more than $2.5\times 10^{11}$ variables has to be solved, which is practically infeasible. Apart from the computational complexity, the memory requirement for storing the optimal policies is also a big challenge. Thus, it is critical to develop alternative approaches to handle $\mathcal{P}1$.

In the next section, we will propose a \textbf{L}yapunov optimization-based \textbf{B}S \textbf{a}ssignment and \textbf{p}ower \textbf{c}ontrol (LBAPC) algorithm  to solve $\mathcal{P}1$, which enjoys the following favorable properties:
\begin{itemize}
\item The decision of the LBAPC algorithm within each time slot is of low complexity, and there is no memory requirement for storing the optimal policy.
\item The LBAPC algorithm has no prior information requirement on the channel statistics or the distribution of the renewable energy processes.
\item The performance of the LBAPC algorithm is controlled by a triplet of control parameters. Theoretically, by adjusting these parameters, the proposed algorithm can behave arbitrarily close to the optimal performance of $\mathcal{P}1$.
\item An upper bound of the required battery capacity is obtained, which shall provide guidelines for practical installation of the EH devices and storage units.
\end{itemize}

\section{Online BS Assignment and Power Control: The LBAPC algorithm}
In this section, we will develop the LBAPC algorithm to solve $\mathcal{P}1$. In order to take the advantage of Lyapunov optimization, we will first introduce a modified NSCM problem to assist the algorithm design. Then the LBAPC algorithm will be proposed for the modified problem, which also provides a feasible solution to $\mathcal{P}1$. In Section V, we will show that this solution is asymptotically optimal for $\mathcal{P}1$.

\subsection{The Modified NSCM Problem}
Due to the energy causality constraint (\ref{EHcausality}), the system's decisions are not independent among different time slots, which makes the design challenging. This is a common difficulty for the design of EH and HES communication systems. {We find that by introducing a non-zero lower bound, $\epsilon_{H_{j}}$, on the battery output power, such coupling effects can be eliminated and the network operations can be optimized by ignoring (\ref{EHcausality}) at each time slot.} Thus, we first introduce a modified version of $\mathcal{P}1$ as follows:
\begin{align}
&\mathcal{P}2:\min_{\mathbf{I}^{t},\mathbf{p}^{t},\mathbf{e}^{t}} \lim_{T\rightarrow +\infty} \frac{1}{T}\sum_{t=0}^{T-1}\mathbb{E}\left[\sum_{k=1}^{K}\left(\tilde{\varphi}_{G} p_{G,k}^{t}\tau+\tilde{\varphi}_{D} I_{D,k}^{t}\right)\right]\nonumber\\
&\ \ \ \ \ \ \ \mathrm{s.t.}\ \ \ (\ref{operconst})-(\ref{EHcausality}), (\ref{channelconst}), (\ref{QoSconst})-(\ref{NonnegaPower})\nonumber\\
&\ \ \ \ \ \ \ \ \ \ \ \ \ \ \sum_{k\in\mathcal{K}}p_{H_{j},k}^{t}\in\Omega_{j},j=1,2,t\in\mathcal{T}\label{tightH2},
\end{align}
where $\Omega_{j}\triangleq \{0\}\bigcup\left[\epsilon_{H_{j}},p_{\mathcal{B}_{j}}^{\max}\right]$ and $0<\epsilon_{H_{j}}\leq p_{\mathcal{B}_{j}}^{\max},j=1,2$.
Compared to $\mathcal{P}1$, an additional constraint on the battery output power is imposed in (\ref{tightH2}), i.e., it is not allowed to be within $\left(0,\epsilon_{H_{j}}\right),j=1,2$. Hence, $\mathcal{P}2$ is a tightened version of $\mathcal{P}1$, and thus any feasible solution for $\mathcal{P}2$ is also feasible for $\mathcal{P}1$. Denote the optimal values of $\mathcal{P}1$ and $\mathcal{P}2$ as $\mathrm{NSC}_{\mathcal{P}1}^{*}$ and $\mathrm{NSC}_{\mathcal{P}2}^{*}$, respectively. The following proposition reveals the relationship between $\mathrm{NSC}_{\mathcal{P}1}^{*}$ and $\mathrm{NSC}_{\mathcal{P}2}^{*}$, which will later help show the asymptotic optimality of the proposed algorithm.
\begin{prop}
The optimal value of $\mathcal{P}2$ is greater than that of $\mathcal{P}1$, but no worse than the optimal value of $\mathcal{P}1$ plus a positive constant $\nu\left(\epsilon_{H_{1}},\epsilon_{H_{2}}\right) $, i.e., $\mathrm{NSC}_{\mathcal{P}1}^{*} \leq \mathrm{NSC}_{\mathcal{P}2}^{*}\leq \mathrm{NSC}_{\mathcal{P}1}^{*}+\nu\left(\epsilon_{H_{1}},\epsilon_{H_{2}}\right)$, where $\nu\left(\epsilon_{H_{1}},\epsilon_{H_{2}}\right)\triangleq \left(1- F_{\mathcal{B}_{1}}^{K}\left(\eta\right)\right)K\tilde{\varphi}_{D}+\epsilon_{H_{2}}\tau\cdot\tilde{\varphi}_{G}
$
and $\eta=\left(2^{\frac{R}{w\tau}}-1\right)\sigma \epsilon_{H_{1}}^{-1}$.
\label{originaltotighten}
\end{prop}
\begin{proof}
See Appendix A.
\end{proof}
In general, the upper bound of $\mathrm{NSC}_{\mathcal{P}2}^{*}$ in Proposition \ref{originaltotighten} is not tight. However, as $\epsilon_{H_{j}},j=1,2$ goes to zero, $\nu\left(\epsilon_{H_{1}},\epsilon_{H_{2}}\right)$ diminishes, as shown in the following corollary.
\begin{corol}
By letting $\epsilon_{H_{1}}$ and $\epsilon_{H_{2}}$ approach zero, $\mathrm{NSC}_{\mathcal{P}2}^{*}$ can be made arbitrarily close to $\mathrm{NSC}_{\mathcal{P}1}^{*}$, i.e.,
$\lim\limits_{\epsilon_{H_{1}},\epsilon_{H_{2}}\rightarrow0}\nu\left(\epsilon_{H_{1}},\epsilon_{H_{2}}\right)=0$.
\label{asyopt1}
\end{corol}
\begin{proof}
First, $\lim\limits_{\epsilon_{H_{2}}\rightarrow0}\epsilon_{H_{2}}\tau\cdot \tilde{\varphi}_{G}=0$. Meanwhile, since
$\lim\limits_{\epsilon_{H_{1}}\rightarrow 0}\eta= +\infty$ and $\lim\limits_{x\rightarrow +\infty}F_{\mathcal{B}_{1}}\left(x\right)=1$,
by combining the two equations, we have $\lim\limits_{\epsilon_{H_{1}}\rightarrow 0}F_{\mathcal{B}_{1}}^{K}\left(\eta\right)=1$,
i.e., $\lim\limits_{\epsilon_{{H}_{1}}\rightarrow 0}\left(1-F_{\mathcal{B}_{1}}^{K}\left(\eta\right)\right)=0$. Thus, the desired result is obtained.
\end{proof}

{Proposition 1 bounds the performance of $\mathcal{P}2$ by that of $\mathcal{P}1$, while Corollary \ref{asyopt1} shows that the performance of both problems can be made arbitrarily close. Actually, Corollary \ref{asyopt1} fits our intuition, since when $\epsilon_{H_{j}}\rightarrow 0$, $\mathcal{P}2$ reduces to $\mathcal{P}1$. Moreover, we see from Proposition 1 that the upper bound of $\mathrm{NSC}_{\mathcal{P}2}^{*}$ equals $\mathrm{NSC}_{\mathcal{P}1}^{*}+\nu\left(\epsilon_{H_{1}},\epsilon_{H_{2}}\right)$, which is a linear function of $\epsilon_{H_{2}}$. As a result, it converges to $\mathrm{NSC}_{\mathcal{P}1}^{*}$ linearly with respect to $\epsilon_{H_{2}}$. Nevertheless, the convergence speed with respect to $\epsilon_{H_{1}}$ depends on the channel statistics.}

\subsection{The LBAPC Algorithm}
In this subsection, we will propose the LBAPC algorithm for $\mathcal{P}2$ based on Lyapunov optimization techniques. It is worth mentioning that the conventional Lyapunov optimization techniques, where the decisions are i.i.d., can not be applied to $\mathcal{P}2$ directly. As mentioned in the last subsection, this is because of the temporal correlation of the battery energy levels which makes the system's decisions time dependent. Fortunately, for the modified NSCM problem, the weighted perturbation method provides an effective solution to circumvent this issue \cite{Neely10CDC}.

To present the algorithm, we will first define the perturbation parameters and the virtual energy queues, which are two critical elements.
\begin{defi}
{The \emph{perturbation parameters} $\theta_{1}$ and $\theta_{2}$ for the EH-BS and HES-BS are bounded constants satisfying}
\begin{equation}
\theta_{1}\geq p_{\mathcal{B}_{1}}^{\max}\tau+\left(VK\tilde{\varphi}_{D}+\bm{1}\{K\neq1\}E_{H_{2}}^{\max}p_{\mathcal{B}_{2}}^{\max}\tau\right)\left(\epsilon_{H_{1}}\tau\right)^{-1},
\label{theta1}
\end{equation}
\begin{equation}
\theta_{2}\geq p_{\mathcal{B}_{2}}^{\max}\tau+\left(VK\tilde{\varphi}_{D}+\bm{1}\{K\neq1\}E_{H_{1}}^{\max}p_{\mathcal{B}_{1}}^{\max}\tau\right)\left(\epsilon_{H_{2}}\tau\right)^{-1},
\label{theta2}
\end{equation}
respectively. In (\ref{theta1}) and (\ref{theta2}), $\bm{1}\{\cdot\}$ denotes the indicator function and $0<V<+\infty$ is a control parameter in the LBAPC algorithm with unit $\mathrm{J}^{2}\cdot \mathrm{cost}^{-1}$. Here, ``$\mathrm{cost}$'' denotes the unit of the network service cost.
\end{defi}

\begin{defi}
The \emph{virtual energy queues} $\tilde{B}_{1}^{t}$ and $\tilde{B}_{2}^{t}$ are defined as
\begin{equation}
\tilde{B}_{j}^{t}=B_{j}^{t}-\theta_{j},j\in\{1,2\},
\label{virtuequeue}
\end{equation}
which are shifted versions of the actual battery energy levels. We denote $\tilde{\bm{B}}^{t}\triangleq\langle \tilde{B}_{1}^{t},\tilde{B}_{2}^{t}\rangle$ for convenience.
\end{defi}

As will be elaborated later, the proposed algorithm minimizes the weighted sum of the virtual queue lengths and the network service cost in each time slot, which shall stabilize $B_{j}^{t}$ around the perturbed energy level $\theta_{j}$ and meanwhile minimize the network service cost. The LBAPC algorithm is summarized in Algorithm 1. In each time slot, the network operations are determined by solving the per-time slot problem, which is parameterized by the current system state, including the CSI and $\tilde{\bm{B}}^{t}$, and with all constraints in $\mathcal{P}2$ except (\ref{EHcausality}). {(A Lyapunov drift-plus-penalty function will be defined in Section V, and one of its upper bounds will be derived, which happens to be the objective function of the per-time slot problem. Also, we will show that ignoring (\ref{EHcausality}) in the per-time slot problem will not affect the feasibility of the LBAPC algorithm for $\mathcal{P}2$.)} This confirms that the LBAPC algorithm does not require prior knowledge of the harvesting processes and the channel statistics, which makes it also applicable for unpredictable environments. We will discuss the solution of the per-time slot problem in the next subsection, and analyze the feasibility as well as the performance of the proposed algorithm in Section V.
\begin{algorithm}[h] 
\caption{The LBAPC Algorithm} 
\label{alg1} 
\begin{algorithmic}[1] 
\STATE At the beginning of each time slot, obtain the virtual energy queue state $\tilde{B}_{j}^{t}$, harvestable energy $\mathcal{E}_{j}^{t}$ and channel gains $h_{\mathcal{B}_{j},k}^{t}$, $\forall k\in\mathcal{K},j=1,2$.
\STATE Decide $\mathbf{e}^{t}, \mathbf{I}^{t}$ and $\mathbf{p}^{t}$ by solving the per-time slot problem, i.e.,
\begin{align}
&\min_{\mathbf{I}^{t},\mathbf{p}^{t},\mathbf{e}^{t}} \sum_{j=1}^{2}\tilde{B}_{j}^{t}\left(e_{j}^{t}-\sum_{k\in\mathcal{K}}p_{H_{j},k}^{t}\tau\right)\nonumber\\
&\ \ \ \ \ \ \ \ \ \ \ \ \ +V\sum_{k=1}^{K}\left(\tilde{\varphi}_{G} p_{G,k}^{t}\tau+\tilde{\varphi}_{D} I_{D,k}^{t}\right)\nonumber\\
&\ \ \ \mathrm{s.t.}\ \ (\ref{operconst}), (\ref{harvestconst}), (\ref{channelconst}), (\ref{QoSconst})-(\ref{tightH2})\nonumber.
\end{align}
\STATE Update the virtual energy queues according to (\ref{batterydynamics}) and (\ref{virtuequeue}).
\STATE Set $t=t+1$.
\end{algorithmic}
\end{algorithm}

\subsection{Solving the Per-Time Slot Problem}
In this subsection, we will develop the optimal solution to the per-time slot problem, which consists of two components: the optimal energy harvesting, i.e., to determine $\mathbf{e}^{t}$, as well as the optimal BS assignment and power control, i.e., to determine $\mathbf{I}^{t}$ and $\mathbf{p}^{t}$. We find the closed-form solution of the optimal energy harvesting decision, and propose an efficient inner-outer optimization (IOO) algorithm to obtain the optimal BS assignment and power control decision. {The results obtained in this subsection are essential for feasibility verification and for performance analysis of the LBAPC algorithm in Section V.}

\textbf{Optimal Energy Harvesting:}
It is straightforward to show that the optimal amount of harvested energy $\mathbf{e}^{t*}$ is obtained by solving the following LP problem:
\begin{align}
&\min_{0\leq e_{j}^{t}\leq \mathcal{E}_{j}^{t}} \sum_{j=1}^{2}\tilde{B}_{j}^{t}e_{j}^{t}
\label{optEHsub}
\end{align}
and its optimal solution is given by
\begin{equation}
e_{j}^{t*}=
\mathcal{E}_{j}^{t}\cdot \bm{1}\{\tilde{B}_{j}^{t}\leq 0\}, j=1,2.
\label{optimalEH}
\end{equation}

\textbf{Optimal BS Assignment and Power Control:} Once decoupling $\mathbf{e}^{t}$ from the objective function, we can then simplify the per-time slot problem into the following optimization problem $\mathcal{P}_{\mathrm{BPAC}}$:
\begin{align}
&\mathcal{P}_{\mathrm{BAPC}}:\min_{\mathbf{I}^{t},\mathbf{p}^{t}} -\sum_{j=1}^{2}\tilde{B}_{j}^{t}\sum_{k\in\mathcal{K}}p_{H_{j},k}^{t}\tau\nonumber\\
&\ \ \ \ \ \ \ \ \ \ \ \ \ \ \ \ \ \ \ \ \ \ +V\sum_{k=1}^{K}\left(\tilde{\varphi}_{G} p_{G,k}^{t}\tau+\tilde{\varphi}_{D} I_{D,k}^{t}\right)\nonumber\\
&\ \ \ \ \ \ \ \ \ \ \ \mathrm{s.t.}\ \ (\ref{operconst}), (\ref{channelconst}), (\ref{QoSconst})-(\ref{tightH2}),\nonumber
\end{align}
which is a difficult mixed integer nonlinear programming (MINLP) problem. For each MU, there are 3 choices of BS assignment, so there are in total $3^{K}$ possible combinations. Moreover, for a fixed $\mathbf{I}^{t}$, a non-convex power control problem should be solved, which further complicates the solution. Thus, the complexity of solving $\mathcal{P}_{\mathrm{BAPC}}$ with exhaustive search is prohibitively high.

We will develop an IOO algorithm to solve $\mathcal{P}_{\mathrm{BAPC}}$, which reduces the search space and avoids solving the non-convex power control problem. In the IOO algorithm, the inner optimization finds the optimal $\{I_{\mathcal{B}_{2},k}^{t}\}$, $\{I_{D,k}^{t}\}$, $\{p_{H_{2},k}^{t}\}$ and $\{p_{G,k}^{t}\}$ for given $\{I_{\mathcal{B}_{1},k}^{t}\}$ and $\{p_{H_{1},k}^{t}\}$, while the outer optimization determines the global optimal $\{I_{\mathcal{B}_{1},k}^{t*}\}$ and $\{p_{H_{1},k}^{t*}\}$.  Denote $\mathcal{H}\triangleq \{k\in\mathcal{K}|I_{\mathcal{B}_{1},k}^{t}=1\}$ and thus, $I_{\mathcal{B}_{2},k}^{t}=0$, $p_{H_{2},k}^{t}=p_{G,k}^{t}=0,\forall k\in\mathcal{H}$.  ($\mathcal{H}^{c}\triangleq\mathcal{K}\setminus\mathcal{H}$.) \footnote{We focus on the non-trivial cases with $\mathcal{H}^{c}\neq \emptyset$, where $\emptyset$ is the empty set.}

\subsubsection{The Inner Problem}
Given $\{I_{\mathcal{B}_{1},k}^{t}\}$ and $\{p_{H_{1},k}^{t}\}$, the inner problem can be written as
\begin{align}
&\mathcal{P}_{\mathrm{in}}:\min_{\{\mathcal{X}_{k}^{t}\}} \sum_{k\in\mathcal{H}^{c}}\left(-\tilde{B}_{2}^{t}p_{H_{2},k}^{t}\tau+V\tilde{\varphi}_{G}p_{G,k}^{t}\tau-V\tilde{\varphi}_{D}I_{\mathcal{B}_{2},k}^{t}\right)\nonumber\\
&\ \ \ \ \ \ \ \mathrm{s.t.}\ \sum_{k\in\mathcal{H}^{c}}I_{\mathcal{B}_{2},k}^{t}\leq N_{\mathcal{B}_{2}},I_{\mathcal{B}_{2},k}^{t}\in\{0,1\},k\in\mathcal{H}^{c}\\
&\ \ \ \ \ \ \ \ \ \ \ \ \ 0\leq p_{\mathcal{B}_{2},k}^{t}\leq I_{\mathcal{B}_{2},k}^{t}p_{\mathcal{B}_{2}}^{\max}, \nonumber\\
&\ \ \ \ \ \ \ \ \ \ \ \ \ \ \ \ \ \ \ \ \ \ \ I_{\mathcal{B}_{2},k}^{t}p_{\mathcal{B}_{2},k}^{t}\geq I_{\mathcal{B}_{2},k}^{t}\rho_{\mathcal{B}_{2},k}^{t}, k\in\mathcal{H}^{c}\\
&\ \ \ \ \ \ \ \ \ \ \ \ \sum_{k\in\mathcal{H}^{c}}p_{\mathcal{B}_{2},k}^{t}\leq p_{\mathcal{B}_{2}}^{\max},\sum_{k\in\mathcal{H}^{c}}p_{H_{2},k}^{t}\in\Omega_{2}\\
&\ \ \ \ \ \ \ \ \ \ \ \ \ p_{H_{2},k}^{t},p_{G,k}^{t}\geq 0, k\in\mathcal{H}^{c},
\end{align}
where $\mathcal{X}_{k}^{t}\triangleq\left[I_{\mathcal{B}_{2},k}^{t},p_{H_{2},k}^{t},p_{G,k}^{t}\right]$ and $\rho_{\mathcal{B}_{2},k}^{t}\triangleq\left(2^{\frac{R}{w\tau}}-1\right)\sigma \slash h_{\mathcal{B}_{2},k}^{t}$ is called the channel inversion power of the HES-BS to $\mathrm{MU}_{k}$ channel. $\mathcal{P}_{\mathrm{in}}$ is obtained by plugging the given $\{I_{\mathcal{B}_{1},k}^{t}\}$ and $\{p_{H_{1},k}^{t}\}$ into the objective function of $\mathcal{P}_{\mathrm{BAPC}}$, utilizing the fact that the transmit power for $\mathrm{MU}_{k}$ should be greater than the channel inversion power, and eliminating $\{I_{D,k}^{t}\}$ with (\ref{operconst}). We denote the optimal value of $\mathcal{P}_{\mathrm{in}}$ as $\mathcal{J}_{\mathrm{in}}\left(\mathcal{H}\right)$.
However, $\mathcal{P}_{\mathrm{in}}$ is still a combinatorial optimization problem and not easy to solve. Fortunately, we find that, without loss of optimality, the MUs with lower channel inversion powers should be assigned to the HES-BS with higher priority, as shown in the following proposition.
\begin{prop}
There exists an optimal solution for $\mathcal{P}_{\mathrm{in}}$ satisfying
\begin{equation}
\rho_{\mathcal{B}_{2},k_{1}}^{t}\leq \rho_{\mathcal{B}_{2},k_{0}}^{t}, \forall k_{1}\in\mathcal{S}_{1}, k_{0}\in\mathcal{S}_{0},
\label{innerproperty1}
\end{equation}
where $\mathcal{S}_{1}\triangleq \{k\in\mathcal{H}^{c}|I_{\mathcal{B}_{2},k}^{t}=1\}$ and  $\mathcal{S}_{0}\triangleq \mathcal{H}^{c}\setminus\mathcal{S}_{1}=\{k\in\mathcal{H}^{c}|I_{\mathcal{B}_{2},k}^{t}=0\}$. In other words, the MUs served by the HES-BS have better channel conditions than the remaining MUs.
\label{increasingCIP}
\end{prop}
\begin{proof}
Suppose for an optimal solution of $\mathcal{P}_{\mathrm{in}}$, $\exists k_{1}\in\mathcal{S}_{1}, k_{0}\in\mathcal{S}_{0}$ such that $\rho_{\mathcal{B}_{2},k_{1}}^{t}>\rho_{\mathcal{B}_{2},k_{0}}^{t}$, we can always serve $\mathrm{MU}_{k_{0}}$ instead of $\mathrm{MU}_{k_{1}}$ with the channel and transmit power that are originally allocated for $\mathrm{MU}_{k_{1}}$, which will not increase the value of the objective function. Hence, there is also an optimal solution for $\mathcal{P}_{\mathrm{in}}$ satisfying (\ref{innerproperty1}), which ends the proof.
\end{proof}
Therefore, we may concentrate on the solutions that satisfy the property in Proposition \ref{increasingCIP}, which means that only the optimal number of MUs that are assigned to the HES-BS, denoted as $m^{*}$, need to be identified. We denote $\left[i\right]$ as the index of the MU in $\mathcal{H}^{c}$ with the $i$th smallest $\rho_{\mathcal{B}_{2},\left[i\right]}^{t}$, and $\mathcal{H}^{c}_{l}$ ($0\leq l\leq |\mathcal{H}^{c}|, \mathcal{H}^{c}_{0}\triangleq\emptyset$) as the set of MUs in $\mathcal{H}^{c}$ with the $l$ smallest $\rho_{\mathcal{B}_{2},\left[l\right]}^{t}$. Besides, $\tilde{N}\triangleq \min\{N_{\mathcal{B}_{2}},|\mathcal{H}^{c}|\}$.

Since $\mathcal{P}_{\mathrm{in}}$ is parameterized by $\tilde{B}_{2}^{t}$, we will investigate the solution of $\mathcal{P}_{\mathrm{in}}$ in the following three disjoint cases: 1) $-\tilde{B}_{2}^{t}>V\tilde{\varphi}_{G}$, 2) $\tilde{B}_{2}^{t}\geq 0$ and 3) $0<-\tilde{B}_{2}^{t}\leq V\tilde{\varphi}_{G}$. {The optimal solution to the inner problem for these three cases will be shown later in Corollary \ref{corolcase1}, \ref{corolcase2} and \ref{corolcase3}, respectively.}

The following lemma reveals a useful property for the solution in case 1).
\begin{lma}
When $-\tilde{B}_{2}^{t}>V\tilde{\varphi}_{G}$, for an optimal solution to $\mathcal{P}_{\mathrm{in}}$, we have $p_{H_{2},k}^{t}=0, \forall k\in\mathcal{H}^{c}$, i.e., no harvested energy in the HES-BS will be consumed.
\label{case1}
\end{lma}
\begin{proof}
Suppose for an optimal solution, $\exists k\in\mathcal{H}^{c}$ such that $p_{H_{2},k}^{t}>0$, it is feasible to construct a new solution with $p_{H_{2},k}^{t'}\!=0$ and $p_{G,k}^{t'}\!=p_{G,k}^{t}+p_{H_{2},k}^{t}$, where the value of the objective function will decrease by $\left(-\tilde{B}_{2}^{t}-V\tilde{\varphi}_{G}\right)p_{H_{2},k}^{t}\tau>\!0$. By contradiction, the result is obtained.
\end{proof}

According to Lemma \ref{case1}, the transmit power of the HES-BS comes from the electric grid, with weight $V\tilde{\varphi}_{G}>0$. Thus, it is optimal to serve the MUs with the channel inversion power. Hence, the optimal solution to $\mathcal{P}_{\mathrm{in}}$ is given by Corollary \ref{corolcase1}.
\begin{corol}
For case 1), i.e., $-\tilde{B}_{2}^{t}>V\tilde{\varphi}_{G}$, the optimal number of MUs assigned to the HES-BS is $m^{*}=
\max\limits_{i\in\mathcal{A}_{\mathrm{case}1}} i \cdot \bm{1}\{\mathcal{A}_{\mathrm{case}1}\neq \emptyset\}$, where $\mathcal{A}_{\mathrm{case}1}=\{i\in\{1,\cdots,\tilde{N}\}\big|
\sum_{k=1}^{i}\rho_{\mathcal{B}_{2},\left[k\right]}^{t}\leq p_{\mathcal{B}_{2}}^{\max},
\tilde{\varphi}_{G}\rho_{\mathcal{B}_{2},\left[i\right]}^{t}\tau\leq \tilde{\varphi}_{D}\}$. And the optimal solution to $\mathcal{P}_{\mathrm{in}}$ is given by
\begin{equation}
\begin{split}
I_{\mathcal{B}_{2},k}^{t}&=\bm{1}\{k\in\mathcal{H}^{c}_{m^{*}}\}\\
p_{H_{2},k}^{t}&=0\\
p_{G,k}^{t}&=I_{\mathcal{B}_{2},k}^{t}\rho_{\mathcal{B}_{2},k}^{t}
\label{solutioncase1}
\end{split},\forall k\in \mathcal{K}.
\end{equation}
\label{corolcase1}
\end{corol}

For case 2), as the weight of the harvested energy consumption is non-positive, i.e., $\tilde{B}_{2}^{t}\geq 0$, we can use $\sum_{k\in\mathcal{H}^{c}}p_{H_{2},k}^{t}=p_{\mathcal{B}_{2}}^{\max}$ to decrease the length of the virtual energy queue by serving as many MUs as possible. Thus, the optimal solution to $\mathcal{P}_{\mathrm{in}}$ can be summarized in Corollary \ref{corolcase2}.
\begin{corol}
For case 2), i.e., $\tilde{B}_{2}^{t}\geq 0$, the optimal number of MUs assigned to the HES-BS is $m^{*}=
\max\limits_{i\in\mathcal{A}_{\mathrm{case}2}} i \cdot \bm{1}\{\mathcal{A}_{\mathrm{case}2}\neq \emptyset\}$, where $\mathcal{A}_{\mathrm{case}2}=\{i\in\{1,\cdots,\tilde{N}\}\big|
\sum_{k=1}^{i}\rho_{\mathcal{B}_{2},\left[k\right]}^{t}\leq p_{\mathcal{B}_{2}}^{\max}\}$. If $m^{*}=0$, which means that all MUs in $\mathcal{H}^{c}$ are experiencing deep fading, the HES-BS is not able to provide service. Otherwise, the optimal solution to $\mathcal{P}_{\mathrm{in}}$ is given by
\begin{equation}
\begin{split}
I_{\mathcal{B}_{2},k}^{t}&=\bm{1}\{k\in\mathcal{H}^{c}_{m^{*}}\}\\
p_{H_{2},k}^{t}&=\frac{\rho_{\mathcal{B}_{2},k}^{t}I_{\mathcal{B}_{2},k}^{t}}{\sum_{k\in\mathcal{H}_{m^{*}}^{c}}\rho_{\mathcal{B}_{2},k}^{t}}p_{\mathcal{B}_{2}}^{\max}\\
p_{G,k}^{t}&=0
\end{split},\forall k\in\mathcal{K}.
\label{solutioncase2}
\end{equation}
\label{corolcase2}
\end{corol}

Denote the total channel inversion powers of the MUs assigned to the HES-BS as $\rho_{\Sigma}$, i.e.,  $\rho_{\Sigma}\triangleq\sum_{k\in\mathcal{H}^{c}}\rho_{\mathcal{B}_{2},k}^{t}I_{\mathcal{B}_{2},k}^{t}$. For case 3), the following result reveals the relationship among $\rho_{\Sigma}$, $\{p_{H_{2},k}^{t}\}$ and $\{p_{G,k}^{t}\}$.
\begin{lma}
When $0<-\tilde{B}_{2}^{t}\leq V\tilde{\varphi}_{G}$, given $\{I_{\mathcal{B}_{2},k}^{t}\}$ ($\sum_{k\in\mathcal{H}^{c}}I_{\mathcal{B}_{2},k}^{t}>0)$\footnote{If $\sum_{k\in\mathcal{H}^{c}}I_{\mathcal{B}_{2},k}^{t}=0$, $p_{H_{2},k}^{t}=p_{G,k}^{t}=0,\forall k\in\mathcal{K}$.}, $\forall k \in \mathcal{K}$
\begin{equation}
\begin{split}
p_{H_{2},k}^{t}&=
\begin{cases}
0,&\rho_{\Sigma}\in\left[0, \frac{-\tilde{B}_{2}^{t}\epsilon_{H_{2}}}{V\tilde{\varphi}_{G}}\right]\\
\frac{\rho_{\mathcal{B}_{2},k}^{t}I_{\mathcal{B}_{2},k}^{t}}{\rho_{\Sigma}}\epsilon_{H_{2}},&\rho_{\Sigma}\in\left( \frac{-\tilde{B}_{2}^{t}\epsilon_{H_{2}}}{V\tilde{\varphi}_{G}},\epsilon_{H_{2}}\right]\\
\rho_{\mathcal{B}_{2},k}^{t}I_{\mathcal{B}_{2},k}^{t},&\rho_{\Sigma}\in\left(\epsilon_{H_{2}},p_{\mathcal{B}_{2}}^{\max}\right]
\end{cases},\\
p_{G,k}^{t}&=
\begin{cases}
\rho_{\mathcal{B}_{2},k}^{t}I_{\mathcal{B}_{2},k}^{t},&\rho_{\Sigma}\in\left[0, \frac{-\tilde{B}_{2}^{t}\epsilon_{H_{2}}}{V\tilde{\varphi}_{G}}\right]\\
0,&\rho_{\Sigma}\in\left( \frac{-\tilde{B}_{2}^{t}\epsilon_{H_{2}}}{V\tilde{\varphi}_{G}},p_{\mathcal{B}_{2}}^{\max}\right]
\end{cases}
\end{split}
.
\label{solutioncase3}
\end{equation}
\label{lemmacase3}
\end{lma}
\begin{proof}
See Appendix B.
\end{proof}

Lemma \ref{lemmacase3} shows that the transmit power comes either from harvested energy or grid energy, i.e., it can not be a mixture of them. Inspired by this property, given the number of MUs served by the HES-BS $i$, we have $I_{\mathcal{B}_{2},k}^{t}=\bm{1}\{k\in\mathcal{H}^{c}_{i}\},\forall k\in\mathcal{K}$, and the optimal power allocation is determined by (\ref{solutioncase3}). Thus, we provide the optimal solution to $\mathcal{P}_{\mathrm{in}}$ in Corollary \ref{corolcase3}.
\begin{corol}
For case 3), i.e., $0<-\tilde{B}_{2}^{t}\leq V \tilde{\varphi}_{G}$, the optimal number of MUs assigned to the HES-BS is $m^{*}=\arg\min\limits_{i\in\mathcal{A}_{\mathrm{case}3}} \sum_{k\in\mathcal{H}^{c}}\left(-\tilde{B}_{2}^{t}p_{H_{2},k}^{t}+V\tilde{\varphi}_{G}p_{G,k}^{t}\right)\tau-V\tilde{\varphi}_{D}i$, where $\mathcal{A}_{\mathrm{case}3}=\{0\}\bigcup
\{i\in\{1,\cdots,\tilde{N}\}\big|
\sum_{k=1}^{i}\rho_{\mathcal{B}_{2},\left[k\right]}^{t}\leq p_{\mathcal{B}_{2}}^{\max}\}$. The optimal $\{I_{\mathcal{B}_{2},k}^{t}\}$ to $\mathcal{P}_{\mathrm{in}}$ is given by $I_{\mathcal{B}_{2},k}^{t}=\bm{1}\{k\in\mathcal{H}_{m^{*}}^{c}\},\forall k\in\mathcal{K}$, and the optimal $\{p_{H_{2},k}^{t}\}$ and $\{p_{G,k}^{t}\}$ can be obtained by using (\ref{solutioncase3}) accordingly.
\label{corolcase3}
\end{corol}

Thus, the inner problem is solved by obtaining $m^{*}$ and the associated power allocation policy by Corollary 2, 3 or 4 depending on the value of $\tilde{B}_{2}^{t}$.

\subsubsection{The Outer Problem}
For the EH-BS, given $\mathcal{H}$, the optimal transmit power is given by
\begin{equation}
p_{H_{1},k}^{t}=
\begin{cases}
0, &\mathcal{H}=\emptyset\\
\frac{\rho_{\mathcal{B}_{1},k}^{t}I_{\mathcal{B}_{1},k}^{t}}{\sum_{k\in\mathcal{H}}\rho_{\mathcal{B}_{1},k}^{t}}p_{\mathcal{B}_{1}}^{\max}, &\mathcal{H}\neq\emptyset, \tilde{B}^{t}_{1}\geq 0\\
\frac{\rho_{\mathcal{B}_{1},k}^{t}I_{\mathcal{B}_{1},k}^{t}}{\min\{\sum_{k\in\mathcal{H}}\rho_{\mathcal{B}_{1},k}^{t},\epsilon_{H_{1}}\}}\epsilon_{H_{1}}, &\mathcal{H}\neq\emptyset, \tilde{B}^{t}_{1}< 0
\end{cases}, k\in\mathcal{H},
\label{EHBSsolution}
\end{equation}
where $\rho_{\mathcal{B}_{1},k}^{t}\triangleq\left(2^{\frac{R}{w\tau}}-1\right)\sigma \slash h_{\mathcal{B}_{1},k}^{t}$. Thus, the outer problem can be formulated as
\begin{align}
\mathcal{P}_{\mathrm{out}}:\min_{\mathcal{H}\in\mathcal{F}_{\mathcal{H}}} \Phi\left(\mathcal{H}\right)+\mathcal{J}_{\mathrm{in}}\left(\mathcal{H}\right)
\end{align}
where
\begin{equation}
\begin{split}
&\Phi\left(\mathcal{H}\right)=\\
&\begin{cases}
0, &\mathcal{H}=\emptyset\\
-\tilde{B}_{1}^{t}p_{\mathcal{B}_{1}}^{\max}\tau-V\tilde{\varphi}_{D}|\mathcal{H}|, &\mathcal{H}\neq\emptyset,\tilde{B}^{t}_{1}\geq 0\\
-\tilde{B}_{1}^{t}\max\{\sum\limits_{k\in\mathcal{H}}\rho_{\mathcal{B}_{1},k}^{t},\epsilon_{H_{1}}\}\tau-V\tilde{\varphi}_{D}|\mathcal{H}|, &\mathcal{H}\neq\emptyset,\tilde{B}^{t}_{1}< 0
\end{cases}
\end{split}
\end{equation}
and $\mathcal{F}_{\mathcal{H}}=\{\mathcal{K}_{s}\subseteq \mathcal{K}\big||\mathcal{K}_{s}|\leq N_{\mathcal{B}_{1}},\sum_{k\in\mathcal{K}_{s}}\rho_{\mathcal{B}_{1},k}^{t}\leq p_{\mathcal{B}_{1}}^{\max}\}$.

The global optimal $\mathcal{H}^{*}$, i.e., $\{I_{\mathcal{B}_{1},k}^{t*}\}$, can be obtained via searching all subsets of $\mathcal{F}_{\mathcal{H}}$, and the associated $\{p_{H_{1},k}^{t*}\}$ can be determined by (\ref{EHBSsolution}). Basically, the IOO algorithm performs a reduced-size search, which eliminates part of the possible combinations of $\mathbf{I}^{t}$'s that are not optimal. More importantly, for fixed $\mathbf{I}^{t}$, closed-form expressions for power allocation are derived to avoid solving the non-convex power control problem. In the worst case, there are $2^{K}$ subsets of $\mathcal{F_{H}}$, and for the inner problem, we have to search $m^{*}$ from 1 to $K$, i.e., the complexity is $\mathcal{O}\left(2^{K}\cdot K\right)$. Such exponential complexity is inevitable in BS assignment problems which are typically NP-hard. However, with a reasonable number of MUs, such complexity is acceptable given the increasing computation power at BSs. Overall, the IOO algorithm brings benefits of accelerating the searching processes compared to exhaustive search, while maintaining the optimality. {In practice, at each time slot, the decision center, i.e., the HES-BS, collects the SI, runs the IOO algorithm and notifies the EH-BS of its decision.} Details of the IOO algorithm are summarized in Algorithm 2.
\begin{algorithm}[h] 
\caption{The IOO Algorithm} 
\label{alg1} 
\begin{algorithmic}[1] 
\STATE Compute the channel inversion power $\rho_{\mathcal{B}_{j},k}^{t}$ based on $h_{\mathcal{B}_{j},k}^{t},j=1,2,k\in\mathcal{K}$.
\STATE Set $\mathcal{F}=\mathcal{F}_{\mathcal{H}}$, $\mathcal{J}^{*}=0$, $I_{D,k}^{t*}=1$, $I_{\mathcal{B}_{1},k}^{t*}=0$, $I_{\mathcal{B}_{2},k}^{t*}=0$, $p_{H_{1},k}^{t*}=p_{H_{2},k}^{t*}=p_{G,k}^{t*}=0$, $\forall k\in\mathcal{K}$.
\STATE \textbf{While} {$\mathcal{F}\neq \emptyset$} \textbf{do}
\STATE \hspace{10pt} Arbitrarily pick $\mathcal{H}\in\mathcal{F}$, set $I_{\mathcal{B}_{1},k}^{t}=\bm{1}\{k\in\mathcal{H}\}$ and\\ \ \ \ \ obtain $\{p_{H_{1},k}^{t}\}$ by (\ref{EHBSsolution}).
\STATE \hspace{10pt} {Based on the value of $\tilde{B}_{2}^{t}$, obtain $\{I_{\mathcal{B}_{2},k}^{t}\}$, $\{I_{D,k}^{t}\}$,\\ \ \ \ \  $\{p_{H_{2},k}^{t}\}$, $\{p_{G,k}^{t}\}$ and the associated $\mathcal{J}_{\mathrm{in}}\left(\mathcal{H}\right)$ with either\\ \ \ \ \  Corollary \ref{corolcase1}, \ref{corolcase2} or \ref{corolcase3}.}
\STATE \hspace{10pt} \textbf{If} {$\Phi\left(\mathcal{H}\right)+\mathcal{J}_{\mathrm{in}}\left(\mathcal{H}\right)<\mathcal{J}^{*}$} \textbf{do}
\STATE \hspace{20pt} $\mathcal{J}^{*}=\Phi\left(\mathcal{H}\right)+\mathcal{J}_{\mathrm{in}}\left(\mathcal{H}\right)$.
\STATE \hspace{20pt} Update $\mathbf{I}^{t*}$, $\mathbf{p}^{t*}$ with $\mathbf{I}^{t}$, $\mathbf{p}^{t}$.
\STATE \hspace{10pt} \textbf{Endif}
\STATE \hspace{10pt} $\mathcal{F}=\mathcal{F}\setminus \mathcal{H}$.
\STATE \textbf{Endwhile}
\end{algorithmic}
\end{algorithm}

\section{Performance Analysis}
One unique advantage of the proposed algorithm is that we can provide theoretical performance analysis and characterize its asymptotic optimality, which will be pursued in this section. We will first prove the feasibility of the LBAPC algorithm for $\mathcal{P}2$, which will be followed by the optimality characterization. Our analysis is based on Lyapunov optimization theory, where the Lyapunov drift function is the key element. During the analysis, an auxiliary optimization problem $\mathcal{P}3$ will be introduced, which bridges the optimal performance of $\mathcal{P}2$ and the performance achieved by the proposed algorithm. Together with Proposition \ref{originaltotighten}, this will establish the asymptotic optimality of the LBAPC algorithm for $\mathcal{P}1$.

We verify the feasibility of the LBAPC algorithm by showing that the battery energy level is confined within a given interval, as demonstrated in the following proposition.
\begin{prop}
Under the LBAPC algorithm, the battery energy level $B_{j}^{t}$ is confined within $\left[0,\theta_{j}+E_{H_{j}}^{\max}\right],j=1,2$.
\label{batteryconfine}
\end{prop}
\begin{proof}
See Appendix C.
\end{proof}

Proposition \ref{batteryconfine} shows that the energy causality constraint (\ref{EHcausality}) will not be violated, which indicates that the LBAPC algorithm is feasible for $\mathcal{P}2$. It also implies that the required battery capacity in the proposed algorithm is $\theta_{j}+E_{H_{j}}^{\max},j=1,2$. In other words, given the size of the available energy storage at the BSs, i.e., $C_{\mathcal{B}_{j}},j=1,2$, we can determine the control parameter $V=\min\{V_{1},V_{2}\}$\footnote{Note that to guarantee $V_{j}>0,j=1,2$, the values of $C_{\mathcal{B}_{j}},j=1,2$ can not be arbitrarily small.} as
\begin{equation}
\begin{split}
V_{j}=\left(K\tilde{\varphi}_{D}\right)^{-1}&\bigg[\left(C_{\mathcal{B}_{j}}-E_{H_{j}}^{\max}-p_{\mathcal{B}_{j}}^{\max}\tau\right)\epsilon_{H_{j}}\tau\\
&-\bm{1}\{K\neq 1\}E_{H_{3-j}}^{\max}p_{\mathcal{B}_{3-j}}^{\max}\tau\bigg],j=1,2.
\end{split}
\end{equation}
{Furthermore, the bounds of the battery energy levels are useful for deriving the main result on the performance of the proposed algorithm.}

Next we will proceed to show the asymptotic optimality of the LBAPC algorithm, for which we first define the Lyapunov function as
\begin{equation}
L\left(\tilde{\bm{B}}^{t}\right)=\frac{1}{2}\sum_{j=1}^{2}\left(B_{j}^{t}-\theta_{j}\right)^{2}=\frac{1}{2}\left[\left(\tilde{B}_{1}^{t}\right)^{2}+
\left(\tilde{B}_{2}^{t}\right)^{2}\right].
\end{equation}
It is a sum of squares of the virtual queue lengths, i.e., the distance between the battery energy levels and the perturbed energy levels. Accordingly, the Lyapunov drift function can be written as
\begin{equation}
\Delta\left(\tilde{\bm{B}}^{t}\right)=\mathbb{E}\left[L\left(\tilde{\bm{B}}^{t+1}\right)-L\left(\tilde{\bm{B}}^{t}\right)|\tilde{\bm{B}}^{t}\right].
\end{equation}
Moreover, the Lyapunov drift-plus-penalty function can be expressed as
\begin{equation}
\Delta_{V}\left(\tilde{\bm{B}}^{t}\right)=\Delta\left(\tilde{\bm{B}}^{t}\right)+V\mathbb{E}\left[\sum_{k\in\mathcal{K}}\left(\tilde{\varphi}_{G}
p_{G,k}^{t}\tau+\tilde{\varphi}_{D}I_{D,k}^{t}\right)|\tilde{\bm{B}}^{t}\right].
\end{equation}
In the following lemma, we derive an upper bound for $\Delta_{V}\left(\tilde{\bm{B}}^{t}\right)$, which will play a critical role throughout the analysis of the LBAPC algorithm.
\begin{lma}
For arbitrary feasible decision variables $\mathbf{e}^{t},\mathbf{I}^{t},\mathbf{p}^{t}$ for $\mathcal{P}2$, $\Delta_{V}\left(\tilde{\bm{B}}^{t}\right)$ is upper bounded by
\begin{equation}
\begin{split}
\Delta_{V}\left(\tilde{\bm{B}}^{t}\right)&\leq \mathbb{E}\Bigg[\sum_{j=1}^{2}\tilde{B}_{j}^{t}\left(e_{j}^{t}-\sum_{k\in\mathcal{K}}p_{H_{j},k}^{t}\tau\right)\\
&+V\sum_{k\in\mathcal{K}}\left(\tilde{\varphi}_{G}
p_{G,k}^{t}\tau+\tilde{\varphi}_{D}I_{D,k}^{t}\right)\big|\tilde{\bm{B}}^{t}\Bigg]+C,
\end{split}
\end{equation}
where $C=\frac{1}{2}\sum_{j=1}^{2}\left(E_{H_{j}}^{\max}\right)^{2}+\frac{1}{2}\sum_{j=1}^{2}\left(p_{\mathcal{B}_{j}}^{\max}\tau\right)^{2}$.
\label{dppupbound}
\end{lma}
\begin{proof}
By subtracting $\theta_{j}$ on both sides of (\ref{batterydynamics}),
\begin{equation}
\tilde{B}_{j}^{t+1}=\tilde{B}_{j}^{t}+e_{j}^{t}-\sum_{k\in\mathcal{K}} p_{H_{j},k}^{t} \tau, j=1,2.
\label{dynamicshift1}
\end{equation}
Squaring both sides of (\ref{dynamicshift1}) and adding up the equalities with $j=1,2$, we have
\begin{equation}
\begin{split}
\sum_{j=1}^{2}\left(\tilde{B}_{j}^{t+1}\right)^{2}=&\sum_{j=1}^{2}\left(\tilde{B}_{j}^{t}+e_{j}^{t}-\sum_{k\in\mathcal{K}} p_{H_{j},k}^{t} \tau\right)^{2}\\
=&\sum_{j=1}^{2}\Bigg[ \left(\tilde{B}_{j}^{t}\right)^{2}+\left(e_{j}^{t}-\sum_{k\in\mathcal{K}}p_{H_{j},k}^{t}\tau\right)^{2}\\
&\ \ \ \ \ \ \ \ \ \ \ \ \ \ \ +2\tilde{B}_{j}^{t}\left(e_{j}^{t}-\sum_{k\in\mathcal{K}}p_{H_{j},k}^{t}\tau\right)\Bigg]\\
\leq&\sum_{j=1}^{2}\Bigg[ \left(\tilde{B}_{j}^{t}\right)^{2}+\left(E_{H_{j}}^{\max}\right)^{2}+\left(p_{\mathcal{B}_{j}}^{\max}\tau\right)^{2}\\
&\ \ \ \ \ \ \ \ \ \ \ \ \ \ \ +2\tilde{B}_{j}^{t}\left(e_{j}^{t}-\sum_{k\in\mathcal{K}}p_{H_{j},k}^{t}\tau\right)\Bigg]
\end{split}.
\label{dynamicshift}
\end{equation}
Dividing both sides of (\ref{dynamicshift}) by $2$, adding $V\sum_{k\in\mathcal{K}}\left(\tilde{\varphi}_{G}p_{G,k}^{t}\tau+\tilde{\varphi}_{D}I_{D,k}^{t}\right)$, as well as taking the expectation conditioned on $\tilde{\bm{B}}^{t}$, we can obtain the desired result.
\end{proof}

Notice that the upper bound of $\Delta_{V}\left(\tilde{\bm{B}}^{t}\right)$ derived in Lemma \ref{dppupbound} coincides with the objective function of the per-time slot problem\footnote{We drop the constant $C$ in the objective function of the per-time slot problem.} in the LBAPC algorithm. To facilitate the performance analysis, we define the following auxiliary problem $\mathcal{P}3$:
\begin{align}
&\mathcal{P}3:\min_{\mathbf{I}^{t},\mathbf{p}^{t},\mathbf{e}^{t}} \lim_{T\rightarrow +\infty} \frac{1}{T}\sum_{t=0}^{T-1}\mathbb{E}\left[\sum_{k=1}^{K}\tilde{\varphi}_{G} p_{G,k}^{t}\tau+\tilde{\varphi}_{D} I_{D,k}^{t}\right]\nonumber\\
&\ \ \ \ \ \ \ \mathrm{s.t.}\ \ \ (\ref{operconst}), (\ref{harvestconst}), (\ref{channelconst}), (\ref{QoSconst})-(\ref{tightH2}),\nonumber\\
&\ \ \ \ \ \ \ \ \ \ \ \ \lim_{T\rightarrow +\infty} \frac{1}{T}\sum_{t=0}^{T-1}\mathbb{E}\left[\sum_{k\in\mathcal{K}}p_{H_{j},k}^{t}\tau-e_{j}^{t}\right]=0,j=1,2\label{EHcausalityrelax}.
\end{align}
In $\mathcal{P}3$, the energy causality constraint (\ref{EHcausality}) in $\mathcal{P}2$ is replaced by (\ref{EHcausalityrelax}), i.e., the average harvested energy consumption equals the average harvested energy. In the following lemma, we will show that $\mathcal{P}3$ is a relaxation of $\mathcal{P}2$.
\begin{lma}
$\mathcal{P}3$ is a relaxation of $\mathcal{P}2$, i.e., $\mathrm{NSC}_{\mathcal{P}3}^{*}\leq \mathrm{NSC}_{\mathcal{P}2}^{*}$, where $\mathrm{NSC}_{\mathcal{P}3}^{*}$ is the optimal value of $\mathcal{P}3$.
\label{relaxPEHcausality}
\end{lma}
\begin{proof}
For any feasible solution of $\mathcal{P}2$, based on the battery dynamics, we have
\begin{equation}
B_{j}^{t+1}=B_{j}^{t}-\sum_{k\in\mathcal{K}}p_{H_{j},k}^{t}\tau +e_{j}^{t},j=1,2,t=0,1,\cdots,T-1.
\end{equation}
Summing up both sides of the above $T$ equalities, taking the expectation, dividing both sides by $T$ and letting $T$ go to infinity, we have
\begin{equation}
\begin{split}
&\lim_{T\rightarrow +\infty}\frac{1}{T}\mathbb{E}\left[B_{j}^{T}\right]=\\
&\lim_{T\rightarrow +\infty}\frac{1}{T}\mathbb{E}\left[B_{j}^{0}\right]
-\lim_{T\rightarrow +\infty}\frac{1}{T}\sum_{t=0}^{T-1}\mathbb{E}\left[\sum_{k\in\mathcal{K}}p_{H_{j},k}^{t}\tau-e_{j}^{t}\right].
\end{split}
\end{equation}
Since $B_{j}^{t}< +\infty$, we have $\lim\limits_{T\rightarrow +\infty}\frac{\mathbb{E}\left[B_{j}^{T}\right]}{T}=0$, i.e., (\ref{EHcausalityrelax}) is satisfied. Hence, any feasible solution of $\mathcal{P}2$ is also feasible to $\mathcal{P}3$, which ends the proof.
\end{proof}

Besides, we find that there exists a stationary and randomized policy \cite{BertsekasDP}, where the decisions are i.i.d. among different time slots, that behaves arbitrarily close to the optimal solution of $\mathcal{P}3$. Meanwhile, the difference between $\mathbb{E}\left[e^{t}_{j}\right]$ and $\mathbb{E}\left[\sum_{k\in\mathcal{K}}p_{H_{j},k}^{t}\tau\right]$ is arbitrarily small. It can be stated mathematically in the following lemma, which will help show the asymptotic optimality of the LBAPC algorithm.
\begin{lma}
For an arbitrary $\delta>0$, there exists a stationary and randomized policy $\Pi$ for $\mathcal{P}3$, which decides $\mathbf{e}^{t\Pi}$, $\mathbf{I}^{t\Pi}$ and $\mathbf{p}^{t\Pi}$, such that (\ref{operconst}), (\ref{harvestconst}), (\ref{channelconst}), (\ref{QoSconst})-(\ref{tightH2}) are met, and the following inequalities are satisfied:
\begin{equation}
\mathbb{E}\left[\sum_{k\in\mathcal{K}}\left(\tilde{\varphi}_{G}p_{G,k}^{t\Pi}\tau+\tilde{\varphi}_{D}I_{D,k}^{t\Pi}\right)\right]\leq \mathrm{NSC}_{\mathcal{P}3}^{*}+\delta, t\in\mathcal{T},
\end{equation}
\begin{equation}
\left|\mathbb{E}\left[\sum_{k\in\mathcal{K}}p_{H_{j},k}^{t\Pi}\tau-e_{j}^{t\Pi}\right]\right|\leq \varrho \delta, j=1,2, t\in\mathcal{T},
\end{equation}
where $\varrho$ is a scaling constant.
\label{veryclose}
\end{lma}
\begin{proof}
The proof can be obtained by Theorem 4.5 in \cite{Neely10}, which is omitted for brevity.
\end{proof}

In Section IV, we bounded the optimal performance of the modified NSCM problem $\mathcal{P}2$ with that of the original NSCM problem $\mathcal{P}1$, while in Lemma \ref{relaxPEHcausality}, we showed the auxiliary problem $\mathcal{P}3$ is a relaxation of $\mathcal{P}2$. With the assistance of these results, next, we will provide the main result in the section, which shows the worst-case performance of the LBAPC algorithm.
\begin{thm}
The network service cost achieved by the LBAPC algorithm, denoted as $\mathrm{NSC}_{\mathrm{LBAPC}}$, is upper bounded by
\begin{equation}
\mathrm{NSC}_{\mathrm{LBAPC}}\leq \mathrm{NSC}^{*}_{\mathcal{P}1} + \nu\left(\epsilon_{H_{1}},\epsilon_{H_{2}}\right)+\frac{C}{V}.
\end{equation}
\label{performancethm}
\end{thm}
\begin{proof}
See Appendix D.
\end{proof}
\begin{rmk}
Theorem \ref{performancethm} implies that the performance of the LBAPC algorithm is controlled by the triplet $\langle \epsilon_{H_{1}},\epsilon_{H_{2}},V\rangle$. By letting $V\rightarrow +\infty, \epsilon_{H_{j}}\rightarrow 0,j=1,2$, the cost upper bound can be made arbitrarily tight, that is, the proposed algorithm asymptotically achieves the optimal performance of the original design problem $\mathcal{P}1$. {Note that in each time slot, the computational complexity of the LBAPC algorithm comes from the IOO algorithm, which is $\mathcal{O}\left(K\cdot 2^{K}\right)$ in the worst case regardless of the choice of $\langle \epsilon_{H_{1}},\epsilon_{H_{2}},V\rangle$.} However, approaching the optimal performance of $\mathcal{P}1$ comes at the expense of a higher battery capacity requirement and longer convergence time to the optimal performance. The reason is that, under the LBAPC algorithm, the battery energy levels will be stabilized at $\theta_{j}$. As $\epsilon_{H_{j}}$ decreases or $V$ increases, $\theta_{j}$ increases accordingly, i.e., the perturbed energy levels are higher, and it will take longer time to accumulate the harvested energy, which postpones the arrival of the network stability and therefore delays the convergence. Thus, by tuning the control parameters, we can achieve different tradeoffs between the system performance and the battery capacity/convergence time.
\end{rmk}

{In this paper, the studied HES network consists of two typical types of renewable energy-powered BSs, and the LBAPC algorithm provides an effective methodology for designing such a network. It is worthwhile noting that the proposed Lyapunov optimization approach can be generalized to HES networks with multiple BSs. However, as the network size increases, the computational complexity of obtaining the optimal solution of the per-time slot problem increases accordingly. Hence, low-complexity algorithms with performance guarantees for the per-time slot problem deserve further investigation.}

\section{Simulation Results}
In this section, we will verify the theoretical results derived in Section V and evaluate the performance of the proposed LBAPC algorithm through simulations. We consider an HES wireless network with 4 MUs unless otherwise specified. In simulations, $\mathcal{E}_{j}^{t}$ is uniformly distributed  between $0$ and $E_{H_{j}}^{\max}$ with the average EH power given by $P_{H_{j}}=E_{H_{j}}^{\max}\left(2\tau\right)^{-1}$, the channel gains are exponentially distributed with mean $g_{0}d^{-4}$, where $g_{0}=-40$ dB is the path-loss constant and $d=50$ m is the distance from the MUs to the BSs. In addition, $\tau=1$ ms, $w_{G}=5$, $R=2$ Kbits\footnote{If a packet is transmitted, the data rate of the link between its destined user and the assigned BS is $R\slash \tau=2$ Mbps in the time slot right after the packet's arrival.}, $w=1$ MHz, $\sigma=10^{-13}$ W, $p_{\mathcal{B}_{1}}^{\max}=p_{\mathcal{B}_{2}}^{\max}=1$ W and $N_{\mathcal{B}_{2}}=K=4$. The costs of the grid energy and packet drop are normalized, i.e., $\varphi_{G}=1$ per Joule and $\varphi_{D}=1$ per packet drop. For comparison, we introduce the Cost-aware Greedy algorithm as a performance benchmark, which gives higher priority to using the harvested energy and optimizes the system cost at the current time slot. It works as follows:
\begin{itemize}
\item First, MUs will be assigned to the EH-BS one by one until the harvested energy is used up. The MUs with higher values of $h_{\mathcal{B}_{1},k}^{t}$ will have a higher priority to be served.
\item Second, the remaining MUs will be assigned to the HES-BS until the harvested energy is used up. The MUs with higher values of $h_{\mathcal{B}_{2},k}^{t}$ will have a higher priority to be served.
\item The remaining MUs will be assigned to the HES-BS one by one if the increment of grid energy cost is less than the decrement of the packet drop cost. Similarly, the MUs with higher $h_{\mathcal{B}_{2},k}^{t}$ will be assigned first.
\end{itemize}

\subsection{Theoretical Results Verification}
\begin{figure}[h]
\begin{center}
    \label{TSCPavg}
   \includegraphics[width=0.45\textwidth]{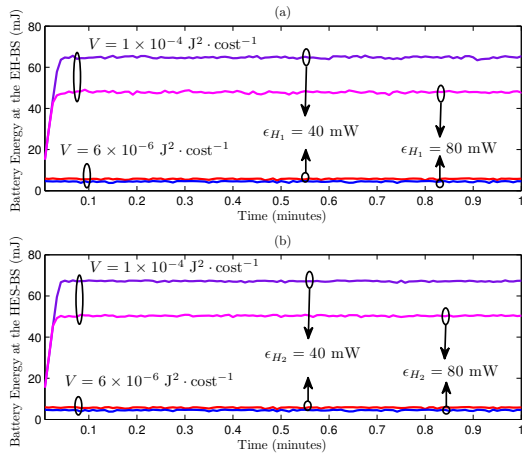}
\end{center}
\vspace{-10pt}
\caption{Battery energy levels, $P_{H_{1}}=P_{H_{2}}=30$ mW and $w_{D}=0.01$.}
\label{BatteryenergyFig}
\end{figure}

In this subsection, we will verify the feasibility and asymptotic optimality of the LBAPC algorithm developed in Proposition \ref{batteryconfine} and Theorem \ref{performancethm}, respectively. {The values of $\theta_{1}$ and $\theta_{2}$ are chosen as the values of the right hand side of (\ref{theta1}) and (\ref{theta2}), respectively.} To verify the feasibility, we show the battery energy levels in Fig. \ref{BatteryenergyFig}. First, we observe that the harvested energy keeps accumulating at the beginning, and finally stabilizes at the perturbed energy levels. The reason is that in the proposed algorithm the Lyapunov drift-plus-penalty function is minimized at each time slot. From the curves, with a larger value of $V$ or a smaller value of $\epsilon_{H_{j}}$, the stabilized energy levels become higher, which coincides with the definition of the perturbation parameters in (\ref{theta1}) and (\ref{theta2}). Also, we see that the energy levels are confined within $\left[0,\theta_{j}+E_{H_{j}}^{\max}\right]$, which verifies Proposition \ref{batteryconfine} and confirms that the energy causality constraint is not violated under the proposed algorithm. The evolution of network service cost with respect to time is shown in Fig. \ref{runingavgCost}. We see that, a larger value of $V$ or a smaller value of $\epsilon_{H_{j}}$ leads to better long-term average performance. Nevertheless, the algorithm converges more slowly to the stable performance. Besides, if $\langle \epsilon_{H_{1}},\epsilon_{H_{2}},V\rangle$ are tuned properly, the proposed algorithm will greatly outperform the Cost-aware Greedy policy.
\begin{figure}[h]
\begin{center}
    \label{TSCPavg}
   \includegraphics[width=0.45\textwidth]{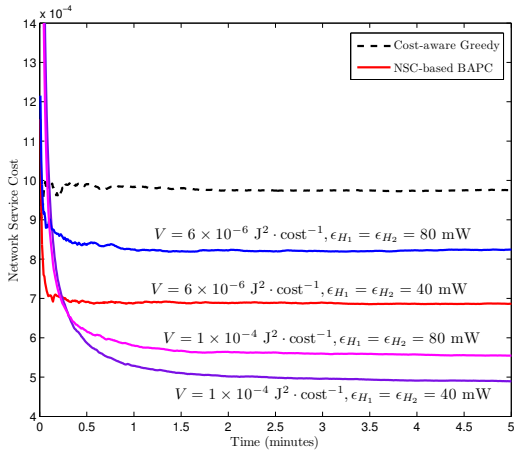}
\end{center}
\vspace{-10pt}
\caption{Network service cost vs. time, $P_{H_{1}}=P_{H_{2}}=30$ mW and $w_{D}=0.01$.}
\label{runingavgCost}
\end{figure}

\begin{figure}[h]
\begin{center}
    \label{TSCPavg}
  \includegraphics[width=0.45\textwidth]{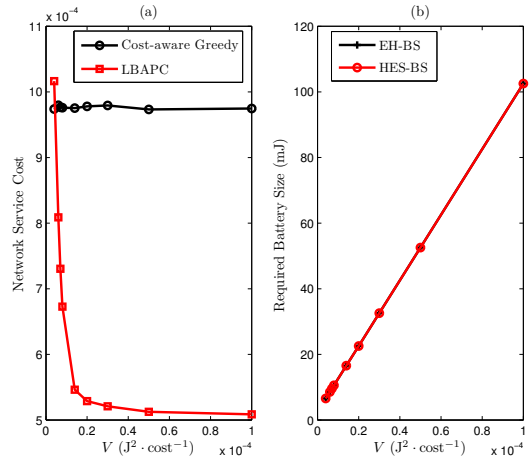}
\end{center}
\vspace{-10pt}
\caption{Network service cost and required battery capacity vs. $V$, $P_{H_{1}}=P_{H_{2}}=30$ mW and $w_{D}=0.01$.}
\label{NSCvsV}
\end{figure}

The relationship between the network service cost\slash required battery capacity and $V$ is shown in Fig. \ref{NSCvsV}. From Fig. \ref{NSCvsV} (a), we see that the network service cost achieved by the proposed algorithm decreases inversely proportional to $V$, and eventually it converges to the optimal value of $\mathcal{P}2$, which verifies Theorem \ref{performancethm}, i.e., the asymptotic optimality. {However, as shown in Fig. \ref{NSCvsV} (b), the required battery capacity grows linearly with $V$, which is because the value of $\theta_{j}$ is linearly increasing with $V$.} Thus, $V$ should be chosen to balance the achievable performance, convergence time and required battery capacity. For instance, if batteries with 100 mJ capacity are available, we can choose $V=1.0\times 10^{-4}\ \mathrm{J}^{2}\cdot \mathrm{cost}^{-1}$ for the LBAPC algorithm, and then $47\%$ performance gain compared to the benchmark will be obtained.

\subsection{Performance Evaluation}
We will show the effectiveness of the proposed algorithm and demonstrate the impacts of various system parameters in this subsection. The relationship between the network service cost and the harvesting power at the EH-BS, i.e., $P_{H_{1}}$, is shown in Fig. \ref{NSCEHpwr}.  We see that the network service cost achieved by either policy is non-increasing with $P_{H_{1}}$, which is in accordance with our intuition since consuming the harvested energy incurs no cost. Also, the LBAPC algorithm significantly reduces the network service cost compared to the greedy algorithm. Besides, the influence of the number of channels at the EH-BS is also revealed. With $N_{\mathcal{B}_{1}}=1$, i.e., in the spectrum limited scenario, increasing $P_{H_{1}}$ brings negligible benefit to the performance in the benchmark policy, since the MU with the best channel condition to the EH-BS is served and this consumes little harvested energy, i.e., $<15$ mW. On the other hand, in the proposed algorithm, the network service cost keeps decreasing with $P_{H_{1}}$ because it minimizes the Lyapunov drift-plus-penalty function at each time slot and the MU being served by the EH-BS is not necessarily the one with the highest channel gain. By increasing $N_{\mathcal{B}_{1}}$ from 1 to 2, the network service cost for both algorithms is greatly reduced, but the Cost-aware Greedy algorithm experiences performance saturation when $P_{H_{1}}\geq 40$ mW. This shows that the proposed algorithm can not only reduce the system cost, but also better utilize the spectrum resource. Meanwhile, increasing $N_{\mathcal{B}_{1}}$ from 2 to 4 has noticeable impact on the greedy algorithm while the benefit to the LBAPC algorithm is minor, which is because the EH-BS becomes energy limited in this region.
\begin{figure}[h]
\begin{center}
    \label{TSCPavg}
   \includegraphics[width=0.45\textwidth]{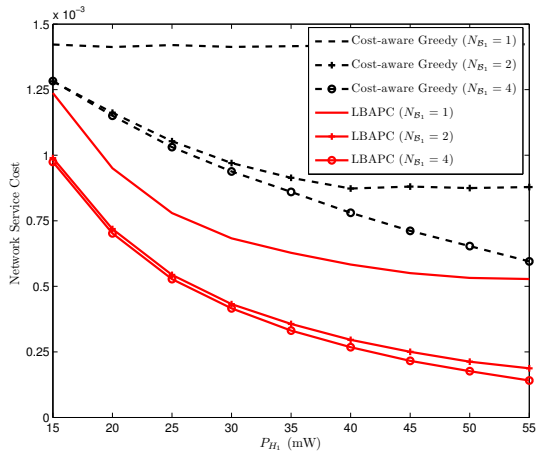}
\end{center}
\vspace{-10pt}
\caption{Network service cost vs. $P_{H_{1}}$, $P_{H_{2}}=30$ mW, $\epsilon_{H_{1}}=\epsilon_{H_{2}}=40$ mW, $C_{\mathcal{B}_{1}}=C_{\mathcal{B}_{2}}=150$ mJ.}
\label{NSCEHpwr}
\end{figure}

The network service cost versus the harvesting power at the HES-BS, i.e., $P_{H_{2}}$, is shown in Fig. \ref{NSCHESpwr}. Similarly, the cost performance decreases as $P_{H_{2}}$ increases. Nevertheless, there is no performance saturation at high $P_{H_{2}}$ since with $N_{\mathcal{B}_{2}}=K$, the MUs can always be served by the HES-BS. From Fig. \ref{NSCHESpwr}, we also see the linear relationship between the network service cost and $P_{H_{2}}$, which is because the harvesting energy and grid energy are co-located at the HES-BS. Additionally, by comparing Fig. \ref{NSCEHpwr} and Fig. \ref{NSCHESpwr}, we observe that the performance improvement by increasing $P_{H_{1}}$ is more obvious than increasing $P_{H_{2}}$, which is due to the diversity gain obtained from independent channels from the two BSs.
\begin{figure}[h]
\begin{center}
    \label{TSCPavg}
   \includegraphics[width=0.45\textwidth]{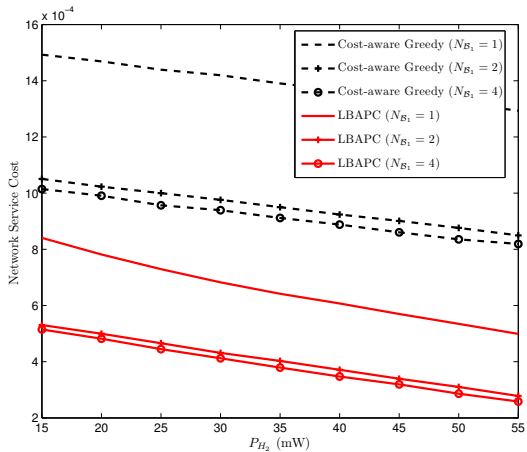}
\end{center}
\vspace{-10pt}
\caption{Network service cost vs. $P_{H_{2}}$, $P_{H_{1}}=30$ mW, $\epsilon_{H_{1}}=\epsilon_{H_{2}}=40$ mW, $C_{\mathcal{B}_{1}}=C_{\mathcal{B}_{2}}=150$ mJ.}
\label{NSCHESpwr}
\end{figure}

The grid energy consumption and packet drop ratio achieved by different algorithms are shown in Fig. \ref{GridQoStradeoff} (a) and Fig. \ref{GridQoStradeoff} (b), respectively. As $w_{D}$ increases, the grid energy consumption increases, meanwhile, the packet drop ratio decreases. Thus, by adjusting the weights of the grid energy cost and packet drop cost, different tradeoffs can be achieved. When $w_{D}$ is sufficiently large, the packet drop ratio achieved by the proposed scheme approaches zero while that achieved by the benchmark remains upon 1.0\%, i.e., the LBAPC algorithm has the potential to meet a higher QoS requirement. From Fig. \ref{GridQoStradeoff}, the proposed algorithm not only outperforms the Cost-aware Greedy algorithm in terms of the network service cost, but it is also more competent to suppress both the grid energy consumption and the packet drops. This indicates that, in HES wireless networks, the optimal energy management should balance the current and future performance, as well as fully utilize the available SI, and the LBAPC algorithm is such a promising solution.
\begin{figure}[h]
\begin{center}
    \label{TSCPavg}
   \includegraphics[width=0.45\textwidth]{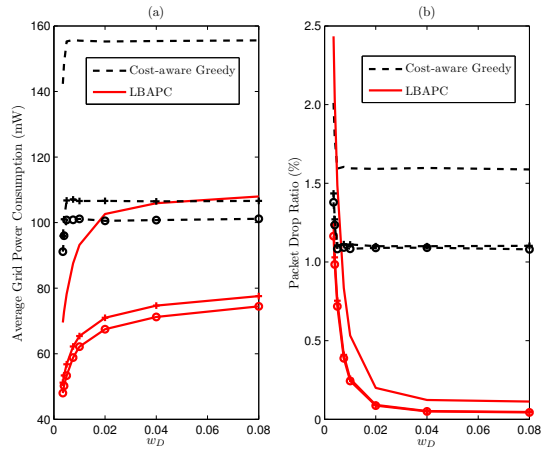}
\end{center}
\vspace{-10pt}
\caption{Grid power consumption and packet drop ratio vs. $w_{D}$, $P_{H_{j}}=30$ mW, $\epsilon_{H_{j}}=40$ mW, $C_{\mathcal{B}_{j}}=150$ mJ, $j=1,2$. The curves with no markers are with $N_{\mathcal{B}_{1}}=1$, while the curves marked with ``$+$'' and ``$\circ$'' are with $N_{\mathcal{B}_{1}}=2$ and $4$, respectively.}
\label{GridQoStradeoff}
\end{figure}

{Finally, we show the relationship between the network service cost and the number of MUs in Fig. \ref{NSCusernum}. We see that the network service cost achieved by either policy is increasing with $K$. This is because the network traffic demand grows as $K$ increases. Due to the insufficient amount of harvested energy, the networks with more MUs are prone to consume more grid energy and drop more packets, both of which contribute to the increase of the network service cost. Besides, the effectiveness of the proposed policy is again validated, and the performance gain compared to the benchmark policy expands as $K$ increases. This highlights the importance of optimal energy management and the benefits of the proposed algorithm for HES cellular networks, especially when the renewable energy resource is scarce, or the traffic load is heavy.}
\begin{figure}[h]
\begin{center}
    \label{TSCPavg}
   \includegraphics[width=0.45\textwidth]{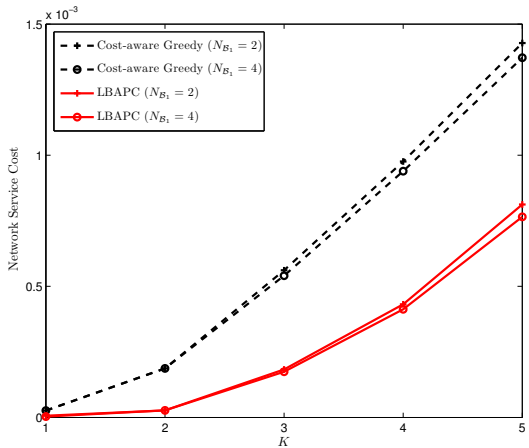}
\end{center}
\vspace{-10pt}
\caption{Network service cost vs. $K$, $P_{H_{1}}=P_{H_{2}}=30$ mW, $N_{\mathcal{B}_{2}}=K$, $\epsilon_{H_{1}}=\epsilon_{H_{2}}=40$ mW, $C_{\mathcal{B}_{1}}=C_{\mathcal{B}_{2}}=150$ mJ.}
\label{NSCusernum}
\end{figure}

\section{Conclusions}
In this paper, we proposed new design methodologies for HES green cellular networks based on Lyapunov optimization techniques. The network service cost, which is comprised of the grid energy cost and the packet drop cost, was adopted as the performance metric, and a BS assignment and power control (BAPC) policy was then developed to optimize the network. A practical online Lyapunov optimization-based BAPC (LBAPC) algorithm was proposed, which requires little prior knowledge and enjoys low computational complexity. Performance analysis was conducted which revealed the asymptotic optimality of the proposed algorithm. Simulation results showed that the proposed LBAPC algorithm significantly outperformed the greedy transmission scheme in terms of the network service cost, grid energy consumption, as well as the achievable QoS. Our results demonstrated the effectiveness of Lyapunov optimization techniques to overcome the curse of dimensionality in MDP solutions, which was the primary barrier in online transmission protocols design for HES green cellular networks. {It will be interesting to extend the proposed algorithm to more general HES networks as well as incorporate more realistic EH models and multi-antenna techniques into consideration, and investigate other design problems such as interference management and user scheduling.}

\begin{appendix}
\subsection{Proof of Proposition \ref{originaltotighten}}
Since $\mathcal{P}2$ is a tightened version of $\mathcal{P}1$, we have $\mathrm{NSC}_{\mathcal{P}1}^{*}\leq \mathrm{NSC}_{\mathcal{P}2}^{*}$. The proof for the other side of the inequality can be obtained by constructing a feasible solution for $\mathcal{P}2$ based on the optimal solution of $\mathcal{P}1$:
\begin{itemize}
\item If $\sum\limits_{k\in\mathcal{K}}p_{H_{1},k}^{t}\in\left(0,\epsilon_{H_{1}}\right)$, then the harvested energy of the EH-BS will not be used in the constructed solution, and the MUs assigned to the EH-BS will experience packet drop.
\item If $\sum\limits_{k\in\mathcal{K}}p_{H_{2},k}^{t}\in\left(0,\epsilon_{H_{2}}\right)$, in the constructed solution, the harvested energy from the HES-BS will be replaced by the grid energy.
\end{itemize}
It is straightforward to verify that the constructed solution is feasible to $\mathcal{P}2$. Thus,
\begin{equation}
\begin{split}
\mathrm{NSC}_{\mathcal{P}2}^{*}&\leq \mathrm{NSC}^{*}_{\mathcal{P}1}+\epsilon_{H_{2}}\tau\cdot \tilde{\varphi}_{G}+\sum_{k=1}^{K}p_{k}\cdot k\tilde{\varphi}_{D}
\\&\leq \mathrm{NSC}^{*}_{\mathcal{P}1}+\epsilon_{H_{2}}\tau\cdot \tilde{\varphi}_{G}+\left(1-p_{0}\right)K{\tilde{\varphi}_{D}},
\end{split}
\label{proof1eq1}
\end{equation}
where $p_{k}$ is the probability that $k$ MUs can be served by the EH-BS with total transmit power $\epsilon_{H_{1}}$. Because of the coupling among the transmit power and the available channels, an exact expression of $p_{0}$ is difficult to obtain. However, we can obtain a lower bound on $p_{0}$ by ignoring the coupling effects, i.e.,
$p_{0}\geq \prod\limits_{k\in\mathcal{K}}\mathbb{P}\left[r\left(h_{\mathcal{B}_{1},k}^{t},\epsilon_{H_{1}}\right)\leq R\right]=F^{K}_{\mathcal{B}_{1}}\left(\eta\right)$.
By substituting the lower bound of $p_{0}$ into (\ref{proof1eq1}), the result is obtained.

\subsection{Proof of Lemma \ref{lemmacase3}}
When $\rho_{\Sigma}\in\left[0, \frac{-\tilde{B}_{2}^{t}\epsilon_{H_{2}}}{V\tilde{\varphi}_{G}}\right]$, suppose there is a solution with $k\in \{k\in\mathcal{H}^{c}|I_{\mathcal{B}_{2},k}^{t}=1\}$, $p_{{H_2},k}^{t}>0$. Due to (\ref{tightH2}), with the solution in (\ref{solutioncase3}) instead, the value of the objective function will decrease by $-\tilde{B}_{2}^{t}\max\big\{\sum\limits_{k\in\mathcal{H}^{c}}p_{H_{2},k}^{t}I_{\mathcal{B}_{2},k}^{t},\epsilon_{H_{2}}\big\}\tau-V\tilde{\varphi}_{G}\rho_{\Sigma}\tau\geq 0$, i.e., $p_{H_{2},k}^{t}=0,\forall k\in\mathcal{K}$ is optimal. Also, as $V\tilde{\varphi}_{G}>0$, it is optimal to transmit with the channel inversion power.

When $\rho_{\Sigma}\in\left( \frac{-\tilde{B}_{2}^{t}\epsilon_{H_{2}}}{V\tilde{\varphi}_{G}},\epsilon_{H_{2}}\right]$, since $-\tilde{B}_{2}^{t}>0$, either $\sum\limits_{k\in\mathcal{H}^{c}}p_{H_{2},k}^{t}=0$ or $\epsilon_{H_{2}}$ is optimal. When $\sum\limits_{k\in\mathcal{H}^{c}}p_{H_{2},k}^{t}=0$, we have $p_{G,k}^{t}=\rho_{\mathcal{B}_{2},k}^{t}I_{\mathcal{B}_{2},k}^{t}$. With the solution in (\ref{solutioncase3}) instead, the value of the objective function will decrease by $V\tilde{\varphi}_{G}\rho_{\Sigma}\tau-\left(-\tilde{B}_{2}^{t}\right)\epsilon_{H_{2}}\tau>0$,
i.e., $\sum\limits_{k\in\mathcal{H}^{c}}p_{H_{2},k}^{t}=\epsilon_{H_{2}}$ is optimal.

When $\rho_{\Sigma}\in\left(\epsilon_{H_{2}},p_{\mathcal{B}_{2}}^{\max}\right]$, by contradiction, the optimal solution should satisfy
\begin{equation}
\begin{cases}
\sum_{k\in\mathcal{H}^{c}}p_{G,k}^{t}I_{\mathcal{B}_{2},k}^{t}&=\lambda\rho_{\Sigma}\\
\sum_{k\in\mathcal{H}^{c}}p_{H_{2},k}^{t}I_{\mathcal{B}_{2},k}^{t}&=\max\{\left(1-\lambda\right)\rho_{\Sigma},\epsilon_{H_{2}}\}
\end{cases},
\end{equation}
where $\lambda\in\left[0,1\right)$. Suppose there is a solution with $k\in\{k\in\mathcal{H}^{c}|I_{\mathcal{B}_{2},k}^{t}=1\}$, $p_{G,k}^{t}>0$, i.e., $\lambda >0$. By using (\ref{solutioncase3}) instead, i.e., $\lambda=0$, the value of the objective function will decrease by
\begin{equation}
\begin{split}
&\ \ \ V\tilde{\varphi}_{G}\lambda\rho_{\Sigma}\tau-\tilde{B}_{2}^{t}\max\{\left(1-\lambda\right)\rho_{\Sigma},\epsilon_{H_{2}}\}\tau-\left(-\tilde{B}_{2}^{t}\right)\rho_{\Sigma}\tau\\
&\geq V\tilde{\varphi}_{G}\lambda\rho_{\Sigma}\tau-\tilde{B}_{2}^{t}\left(1-\lambda\right)\rho_{\Sigma}\tau-\left(-\tilde{B}_{2}^{t}\right)\rho_{\Sigma}\tau\\
&=
\left(V\tilde{\varphi}_{G}-\left(-\tilde{B}_{2}^{t}\right)\right)\rho_{\Sigma}\lambda\tau\geq 0,
\end{split}\nonumber
\end{equation}
i.e., $\lambda=0$ is optimal. Similarly, as $-\tilde{B}_{2}^{t}>0$, it is optimal to transmit with the channel inversion power. To summarize, (\ref{solutioncase3}) is optimal and thus Lemma \ref{lemmacase3} is obtained.

\subsection{Proof of Proposition \ref{batteryconfine}}
We shall first prove that $B_{j}^{t}$ is upper bounded by $\theta_{j}+E_{H_{j}}^{\max},j=1,2$, based on the optimal harvested energy in (\ref{optimalEH}). Suppose $\theta_{j} \leq B_{j}^{t}\leq \theta_{j}+E_{H_{j}}^{\max}$, and since $e_{j}^{t*}=0$, we have $B_{j}^{t+1}\leq B_{j}^{t}\leq \theta_{j}+E_{H_{j}}^{\max}$. Otherwise, if $0\leq B_{j}^{t}<\theta_{j}$, since $e_{j}^{t*}=\mathcal{E}_{j}^{t}$, we have $B_{j}^{t+1}\leq B_{j}^{t}+e_{j}^{t*}<\theta_{j}+E_{H_{j}}^{\max}$.

Next, we reveal an important property of the optimal solution of the per-time slot problem to assist the proof for the lower bound. It indicates that if the battery energy level at a BS is below a threshold, no harvested energy in this BS will participate in transmission, as shown in the following lemma.
\begin{lma}
If $B_{j}^{t}< p_{\mathcal{B}_{j}}^{\max}\tau$, then $p_{H_{j},k}^{t*}=0, j=1,2, k\in\mathcal{K}$, where $p_{H_{j},k}^{t*}$ is the optimal battery output power obtained in the per-time slot problem.
\label{lmazeropwr}
\end{lma}
\begin{proof}
We start with the single user case and omit the user index. Suppose when $B_{j}^{t}< p_{\mathcal{B}_{j}}^{\max}\tau$, there is an optimal solution where $p_{H_{j}}^{t}>0$ ($I_{\mathcal{B}_{j}}^{t}=1,I_{\mathcal{B}_{3-j}}^{t}=I_{D}^{t}=0$). Since $p_{H_{j}}^{t}\in\Omega_{j}$, thus $p_{H_{j}}^{t}\geq \epsilon_{H_{j}}$ . With this solution, the value of the objective function in the per-time slot problem is no less than
$-\tilde{B}_{j}^{t}\epsilon_{H_{j}}\tau > V\tilde{\varphi}_{D}$, which is greater than what is achieved by the solution with $I_{D}^{t}=1$, $p_{H_{j}}^{t}=p_{H_{3-j}}^{t}=p_{G}^{t}=0$, i.e., $p_{H_{j}}^{t*}=0$. For the multi-user scenario, again, suppose when $B_{j}^{t}<p_{\mathcal{B}_{j}}^{\max}\tau$, there is an optimal solution where $\sum_{k\in\mathcal{K}}p_{H_{j},k}^{t}>0$, i.e., $\sum_{k\in\mathcal{K}}p_{H_{j},k}^{t}\geq\epsilon_{H_{j}}$. With this solution, the minimum value of the objective function is $-\tilde{B}_{j}^{t}\epsilon_{H_{j}}\tau - E_{H_{3-j}}^{\max}p_{\mathcal{B}_{3-j}}^{\max}\tau$, which is achieved when $B_{3-j}^{t}=\theta_{3-j}+E_{H_{3-j}}^{\max}$, meanwhile all the MUs are served with the harvested energy, and at least one of them is served by BS $\mathcal{B}_{3-j}$. With the definition of $\theta_{j}$, we are able to show $-\tilde{B}_{j}^{t}\epsilon_{H_{j}}\tau - E_{H_{3-j}}^{\max}p_{\mathcal{B}_{3-j}}^{\max}\tau>KV\tilde{\varphi}_{D}$, which is achieved by the solution with $I_{D,k}^{t}=1$, $\forall k\in\mathcal{K}$, and $\sum_{k\in\mathcal{K}}p_{H_{j},k}^{t}=0$, i.e., $p_{H_{j},k}^{t*}=0,\forall k\in\mathcal{K}$.
\end{proof}

Based on the property of $p_{H_{j},k}^{t*}$ derived in Lemma \ref{lmazeropwr}, now we proceed to show $B_{j}^{t}$ is lower bounded by zero. Suppose $0\leq B_{j}^{t}< p_{\mathcal{B}_{j}}^{\max}\tau$, according to Lemma \ref{lmazeropwr}, $p_{H_{j},k}^{t*}=0,\forall k\in\mathcal{K}$, thus $B_{j}^{t+1}\geq B_{j}^{t}\geq 0$. Otherwise, if $B_{j}^{t}>p_{\mathcal{B}_{j}}^{\max}\tau$, with (\ref{PPC}), $\sum\limits_{k\in\mathcal{K}}p_{H_{j},k}^{t*}\leq p_{\mathcal{B}_{j}}^{\max}$, thus $B_{j}^{t+1}\geq B_{j}^{t}-\sum\limits_{k\in\mathcal{K}}p_{H_{j},k}^{t*}\tau\geq 0$. As a result, $B_{j}^{t}\in\left[0,\theta_{j}+E_{H_{j}}^{\max}\right]$.

\subsection{Proof of Theorem \ref{performancethm}}
Since the LBAPC algorithm obtains the optimal solution of the per-time slot problem,
\begin{equation}
\begin{split}
\Delta_{V}\left(\tilde{\bm{B}}^{t}\right)&\leq  \mathbb{E}\Bigg[\sum_{j=1}^{2}\tilde{B}_{j}^{t}\left(e^{t*}_{j}-\sum_{k\in\mathcal{K}}p_{H_{j},k}^{t*}\tau\right)\\
&+V\sum_{k\in\mathcal{K}}\left(\tilde{\varphi}_{G}p_{G,k}^{t*}\tau+\tilde{\varphi}_{D}I_{D,k}^{t*}\right)\big|\tilde{\bm{B}}^{t}\Bigg]+C\\
&\leq \mathbb{E}\Bigg[\sum_{j=1}^{2}\tilde{B}_{j}^{t}\left(e^{t\Pi}_{j}-\sum_{k\in\mathcal{K}}p_{H_{j},k}^{t\Pi}\tau\right)\\
&+V\sum_{k\in\mathcal{K}}\left(\tilde{\varphi}_{G}p_{G,k}^{t\Pi}\tau+\tilde{\varphi}_{D}I_{D,k}^{t\Pi}\right)\big|\tilde{\bm{B}}^{t}\Bigg]+C\\
&=\mathbb{E}\left[\sum_{j=1}^{2}\tilde{B}_{j}^{t}\left(e^{t\Pi}_{j}-\sum_{k\in\mathcal{K}}p_{H_{j},k}^{t\Pi}\tau\right)\big|\tilde{\bm{B}}^{t}\right]\\
&+V\mathbb{E}\left[\sum_{k\in\mathcal{K}}\left(\tilde{\varphi}_{G}p_{G,k}^{t\Pi}\tau+\tilde{\varphi}_{D}I_{D,k}^{t\Pi}\right)\right]+C\\
&\mathop{\leq}^{(\dagger)} \varrho \delta \sum_{j=1}^{2} \max\{\theta_{j},E_{H_{j}}^{\max}\} + V \left(\mathrm{NSC}_{\mathcal{P}3}^{*}+\delta\right)+C,
\end{split}
\end{equation}
where $(\dagger)$ is due to Lemma \ref{veryclose}. By letting $\delta$ go to zero, we obtain
\begin{equation}
\Delta_{V}\left(\tilde{\bm{B}}^{t}\right)\leq V\mathrm{NSC}_{\mathcal{P}3}^{*}+C.
\label{thmeq1}
\end{equation}
Taking the expectation on both sides of (\ref{thmeq1}), summing up the equations for $t=0,\cdots T-1$, dividing by $T$ and letting $T\rightarrow +\infty$, we have
$\mathrm{NSC}_{\mathrm{LBAPC}}\leq \mathrm{NSC}_{\mathcal{P}3}^{*}+\frac{C}{V}$.
By further utilizing Proposition \ref{originaltotighten} and Lemma \ref{relaxPEHcausality}, the theorem is proved.
\end{appendix}

\end{document}